\documentclass[11pt,draftcls,onecolumn]{IEEEtran}

\usepackage{color}
\usepackage{mdwlist}
\usepackage{subfigure}
\usepackage{setspace}
\usepackage{comment}
\usepackage{booktabs}
\usepackage{graphicx}
\usepackage{amsthm,amsfonts,amsmath,amssymb}
\usepackage[numbers]{natbib}
\usepackage{multirow}
\usepackage{tabulary}
\usepackage{algorithm}
\usepackage{algorithmic}
\usepackage{enumerate}

\newcommand{\eop}{\hfill $\blacksquare$}

\newcommand{\complex}{\mathbb{C}}

\newcommand{\zeronorm}[1]{||{#1}||_0}

\newcommand{\cnormal}{{\cal CN}}

\newcommand{\signal}{\vec{x}}

\newcommand{\transform}{\vec{X}}

\newcommand{\msignal}{\bf{x}}
\newcommand{\msignalc}[2]{x[{#1}][#2]}
\newcommand{\mtransform}{\bf{X}}
\newcommand{\mtransformc}[2]{X[{#1}][{#2}]}

\newcommand{\rr}[1]{r_{x, #1}}
\newcommand{\rc}[1]{r_{y, #1}}
\newcommand{\sr}[1]{s_{{#1}}}
\newcommand{\pprimef}[1]{{\cal Q}_{#1}}

\newcommand{\binobsv}[2]{\vec{y}_{b,#1,#2}}

\newcommand{\stages}{d}

\newcommand{\delays}{D}

\newcommand{\osratio}{r}

\newcommand{\sparsity}{k}
\newcommand{\sindex}{\delta}
\newcommand{\samples}{m}

\newcommand{\binsize}{\mathbf{F}} 

\newcommand{\primef}[1]{{\cal P}_{#1}}
\newcommand{\primefr}[1]{{\cal P}_{x,#1}}
\newcommand{\primefc}[1]{{\cal P}_{y,#1}}

\newcommand{\subr}[1]{{n}_{x, #1}}
\newcommand{\subc}[1]{{n}_{y, #1}}

\newtheorem{theorem}{Theorem}[section]

\begin{document}
\title{Fast and Efficient Sparse 2D Discrete Fourier Transform using Sparse-Graph Codes}
\author{
Frank Ong, Sameer Pawar and Kannan Ramchandran
}

\author{
    \IEEEauthorblockN{Frank Ong, Sameer Pawar, and Kannan Ramchandran}\\
    \IEEEauthorblockA{Department of Electrical Engineering and Computer Sciences\\
University of California, Berkeley
    \\\{frankong, spawar, kannanr\}@eecs.berkeley.edu}
}

\maketitle
\begin{abstract}

We present a novel algorithm, named the 2D-FFAST \footnote{Two-dimensional Fast Fourier Aliasing-based Sparse Transform}, to compute a sparse 2D-Discrete Fourier Transform (2D-DFT) featuring both low sample complexity and low computational complexity. The proposed algorithm is based on mixed concepts from signal processing (sub-sampling and aliasing), coding theory (sparse-graph codes) and number theory (Chinese-remainder-theorem) and generalizes the 1D-FFAST \footnote{Fast Fourier Aliasing-based Sparse Transform} algorithm recently proposed by Pawar and Ramchandran \cite{Pawar:2013da} to the 2D setting. Concretely, our proposed 2D-FFAST algorithm computes a  $\sparsity$-sparse 2D-DFT, with a uniformly random support, of size $N = N_x \times N_y$ using $O(\sparsity)$ noiseless spatial-domain measurements in $O(\sparsity\log\sparsity)$ computational time. Our results are attractive when the sparsity is sub-linear with respect to the signal dimension, that is, when $\sparsity \rightarrow \infty$ and $\sparsity / N \rightarrow 0$. For the case when the spatial-domain measurements are corrupted by additive noise, our 2D-FFAST framework extends to a noise-robust version in sub-linear time of $O(\sparsity\log^{4}N)$ using $O(\sparsity\log^{3}N)$ measurements. Simulation results, on synthetic images as well as real-world magnetic resonance images, are provided in Section~\ref{sec:simulations} and demonstrate the empirical performance of the proposed 2D-FFAST algorithm.

\end{abstract}

\begingroup
\let\clearpage\relax
\section{Introduction}

In many imaging applications, such as magnetic resonance angiography, computed tomography, and astronomical imaging, the image of interest has a sparse representation in the Fourier domain. Recent results in compressed sensing~\cite{Donoho:2006ci, Candes:2006eq} allow us to exploit this sparse structure to acquire and reconstruct signals from far fewer measurements than required by the Shannon-Nyquist theorem. However, the most popular class of compressed sensing based reconstruction algorithms involves alternating iteratively between the spatial domain representation and the Fourier domain representation of the signals, and consequently are computationally expensive. Hence, such algorithms have limited scope in devices and acquisition systems demanding inexpensive, low-power or real-time signal analysis. 

While many algorithms with low computational complexity \cite{Gilbert:2002dw, Gilbert:2005gg, Gilbert:2008ba, Gilbert:dga, Hassanieh:2012ey, Hassanieh:2012uq, Indyk:2014tm, Indyk:2014if} have been proposed to compute a sparse 1D discrete-Fourier-transform (DFT), extensions of these algorithms to 2D sparse signals, with recovery guarantees similar to 1D, are non-trivial and very few 2D algorithms \cite{Indyk:2014if, Ghazi:2013ek} have been proposed. Yet multidimensional signals, such as images and videos, often have much sparser representations than 1D signals and can be found in a wide range of sparse signal processing applications. Many imaging applications further require the ability to reconstruct the sparse image using only a few measurements in the Fourier domain so as to reduce imaging costs, such as acquisition time in magnetic resonance imaging or radiation dose in computed tomography. Hence, a practically feasible algorithm with \emph{both low computational complexity and low sample complexity} for computing a sparse 2D-DFT is of great interest.

\begin{figure}[h]
\includegraphics[width=\linewidth]{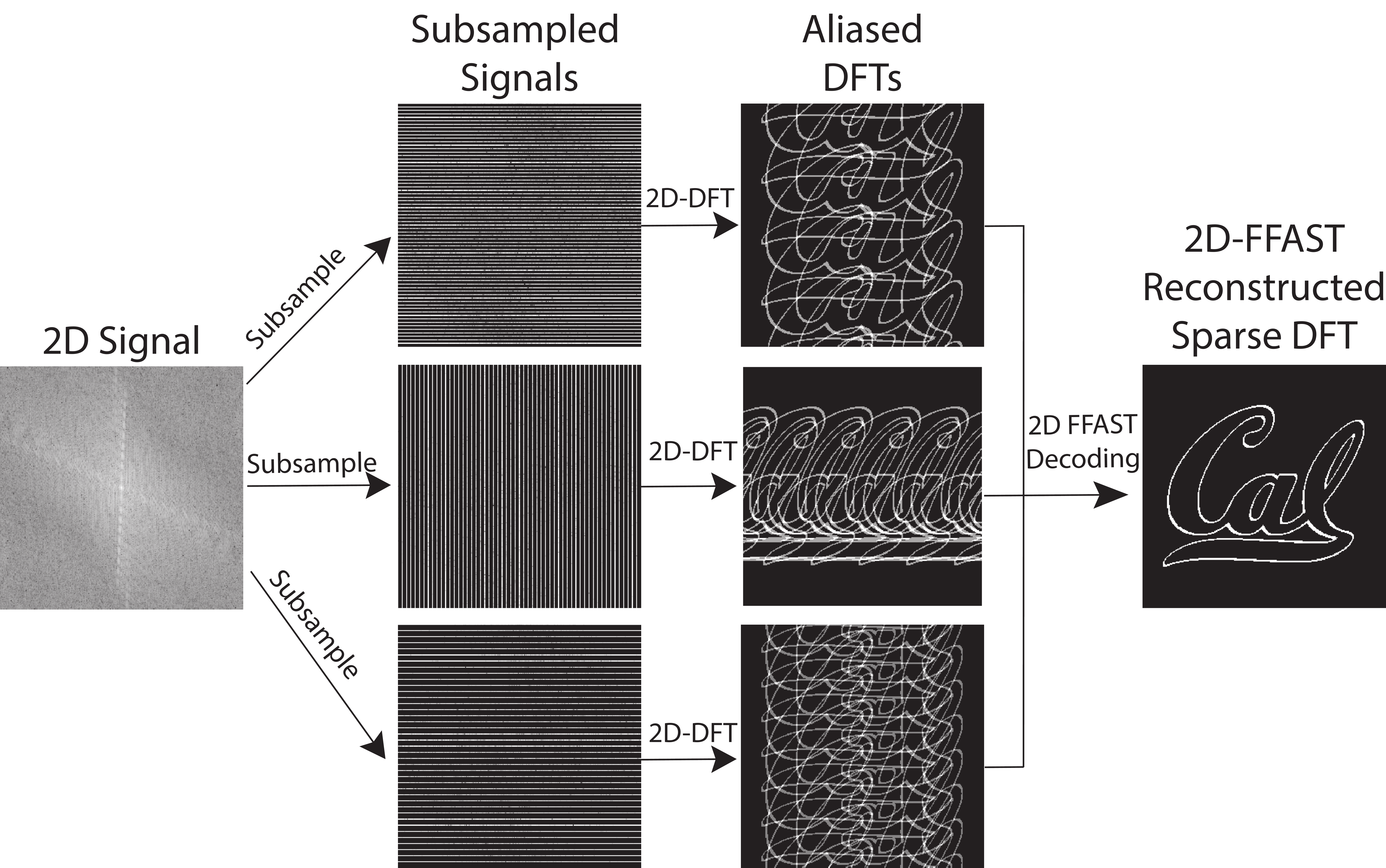}
\caption{A simplified visual illustration of how the 2D-FFAST algorithm computes a 2D sparse DFT. The 2D-FFAST algorithm uniformly subsamples the 2D signal with different sampling factors, thereby producing different aliasing patterns in the DFT domain. Each aliased DFT provides unique and independent information about the sparse DFT coefficient  and can be viewed as a sparse graph. The 2D-FFAST algorithm uses all the aliased 2D-DFTs together to reconstruct the full 2D spectrum via an onion-peeling style algorithm as described in Section \ref{sec:extension}.}
\label{fig:illustration}
\end{figure}

In this work, we present an algorithm, named the 2D-FFAST (Fast Fourier Aliasing-based Sparse Transform), to compute a sparse 2D-DFT with both low sample complexity and computational complexity. The proposed algorithm generalizes the 1D-FFAST approach in Pawar and Ramchandran~\cite{Pawar:2013da} to the 2D setting in two cases: the first case is when the 2D signal dimensions $N_x$ and $N_y$ are co-prime, that is their greatest common divisor (gcd) is $1$. In this case, there is a unique one-to-one mapping between the 1D-DFT and 2D-DFT and in Section~\ref{sec:good} we show that the 1D-FFAST described in Pawar and Ramchandran~\cite{Pawar:2013da} can be used in a straightforward manner. 

Our key contribution is the design of the 2D-FFAST algorithm for the second and more general case, in which we consider a wider class of 2D signals as described in Section \ref{sec:model}. Although in this case, the 2D-DFT cannot be trivially mapped into a 1D-DFT, we show that the 1D-FFAST architecture can still be lifted to the 2D setting as shown in Section \ref{sec:extension} and hence the recovery guarantees are preserved as well. 

At a high level, the 2D FFAST algorithm induces sparse graph codes in the 2D-DFT domain via a Chinese-Remainder-Theorem (CRT)-guided 2D sub-sampling operation in the spatial-domain. The resulting sparse graph codes are then exploited to devise a fast and iterative onion-peeling style algorithm that computes 2D-DFT, while achieving both low sample complexity and low computational complexity. A simplified visual illustration of the mechanics of the 2D-FFAST architecture is provided in Figure \ref{fig:illustration} and \ref{fig:illustration2}. 

Performance-wise, for the case when the signal-domain samples are observed without any noise, our proposed 2D-FFAST algorithm computes a $\sparsity$-sparse 2D-DFT of size $N = N_x \times N_y$, using $O(\sparsity)$ measurements in $O(\sparsity\log\sparsity)$ computations for $\sparsity = o(N)$, that is when $\sparsity \rightarrow \infty$ and $\sparsity / N \rightarrow 0$. Further, when the observed spatial-domain samples are corrupted by additive white Gaussian noise, the 2D-FFAST framework computes a $\sparsity$-sparse 2D-DFT in sub-linear time of $O(\sparsity\log^{4}N)$ using $O(\sparsity\log^{3}N)$ measurements.

Before diving into the main results, we emphasize the following caveats of our algorithm and analysis: 

\begin{enumerate}
\item Our proposed 2D-FFAST algorithm does not apply to arbitrary 2D signal dimensions but only to a selected but arbitrarily large set of dimensions as described in Section \ref{sec:model}. However, as we have demonstrated in the same section, two dimensional choices that are close to a user desired dimension and satisfy the 2D-FFAST requirement can be found.
\item Our analytical results are probabilistic and hold for asymptotic values of $\sparsity$ and signal dimensions $N_x \times N_y$, with a success probability that approaches 1 asymptotically. 
\item Our analytical results also assume a uniformly random model for the support of the non-zero DFT coefficients, which does not describe real-world signal accurately. 
\item For the noisy case, we assume that the non-zero DFT coefficients belong to an arbitrarily large but finite constellation such that the effective signal-to-noise ratio is finite for analysis purpose. However, as we will show in the simulation section (Section~\ref{sec:simulations}), our 2D-FFAST simulation results on synthetic images as well as real world magnetic resonance images (MRI) show good empirical performance even though though these signals have a non-uniform support for the dominant DFT coefficients that further take values from arbitrary complex constellations.
\end{enumerate}

\begin{figure}[h]
\includegraphics[width=\linewidth]{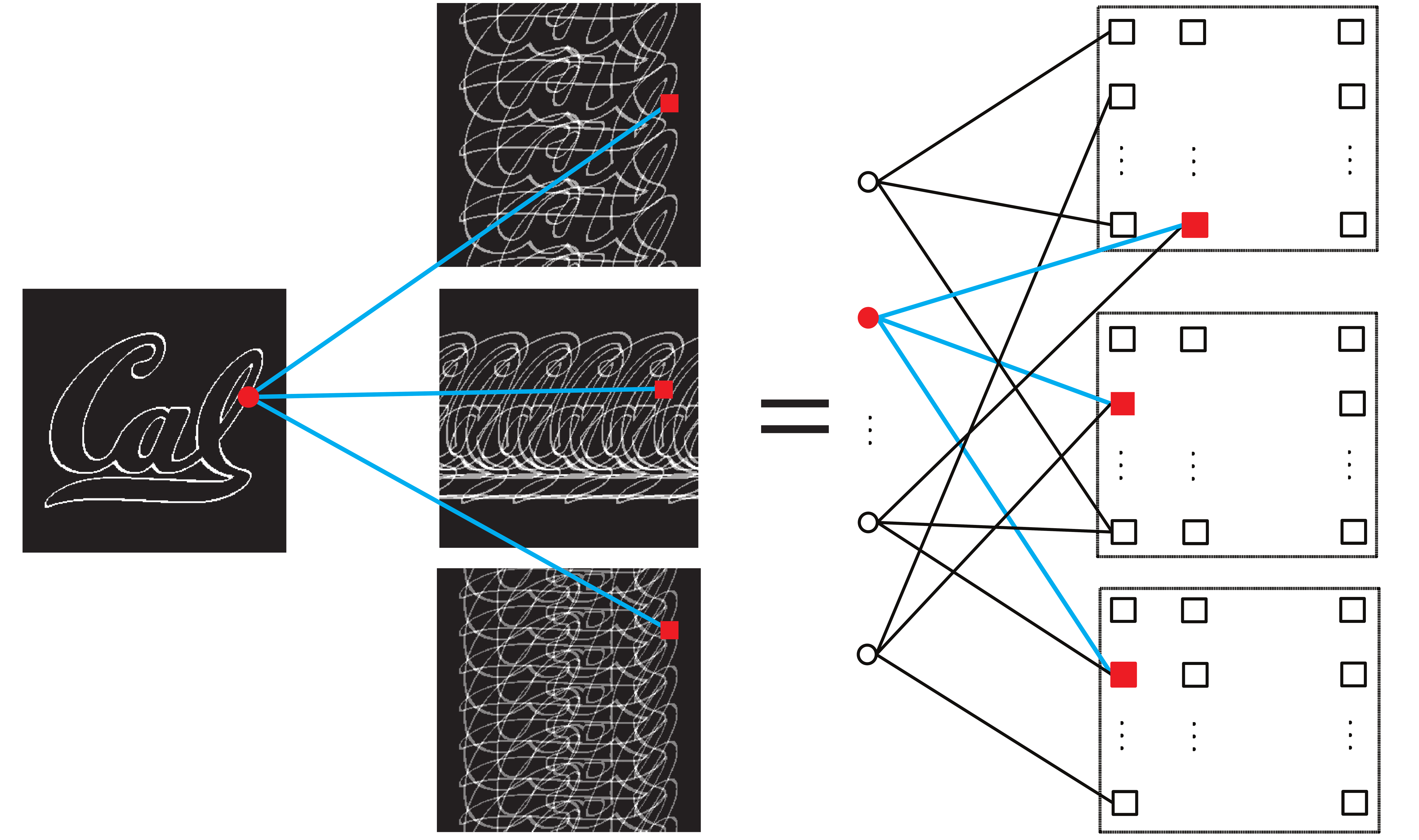}
\caption{A visual illustration of how the 2D-FFAST algorithm transforms different aliasing patterns into a balls-and-bins sparse graph. This effectively transforms the 2D sparse DFT problem into a sparse graph decoding problem.}
\label{fig:illustration2}
\end{figure}

\subsection{Related Work}\label{sec:related}
A number of previous works \cite{Gilbert:2002dw, Gilbert:2005gg, Gilbert:2008ba, Gilbert:dga, Hassanieh:2012ey, Hassanieh:2012uq, Indyk:2014tm, Pawar:2013da} have addressed the problem of computing a 1D DFT of an $ N$-length signal that has a $\sparsity$-sparse Fourier transform, in sub-linear time and sample complexity. Most of these algorithms achieve a sub-linear time performance by first isolating the non-zero DFT coefficients into different groups/bins, using specific filters. The non-zero DFT coefficients are then recovered iteratively, one at a time. The filters used for the binning operation are typically of length $O(\sparsity\log N)$. As a result, the sample and the computational complexity is $O(\sparsity\log N)$ or more even for noiseless measurements. The 1D-FFAST algorithm in Pawar and Ramchandran \cite{Pawar:2013da} is the first algorithm to compute a $\sparsity$-sparse 1D-DFT of an $ N$-length signal using $O(\sparsity)$ measurements and $O(\sparsity\log\sparsity)$ computations for noiseless measurements. The 1D-FFAST algorithm is further extended in \cite{Pawar:2014gx} to additive Gaussian noise corrupted measurements using $O(\sparsity\log^{3} N)$ measurements.

Unlike 1D-DFT, there are very few algorithms designed for multi-dimensional DFT. The algorithm in Gilbert et al. \cite{Gilbert:2005gg} achieves $O(\sparsity\log^c N)$ sample and time complexity for computing a $\sparsity$-sparse $N=N_x\times N_y$ 2D-DFT, where $c$ is some positive constant. In Ghazi et al. \cite{Ghazi:2013ek}, the authors propose a sub-linear time algorithm for computing a $\sparsity$-sparse 2D-DFT for 2D signals with equal dimension with $N_x = N_y = \sqrt{N}$, where the dimensions are powers of $2$. For noiseless measurements with $\sparsity = O(\sqrt{N})$, the algorithm uses $O(\sparsity)$ measurements and $O(\sparsity\log\sparsity)$ time. For the general sub-linear sparsity regime the computational complexity is $O(\sparsity\log\sparsity + \sparsity(\log\log N)^c)$ for some constant $c$. The algorithm also succeeds with a {\em constant probability} that does not approach $1$, which generally translates to inferior empirical results as well. In Indyk et al. \cite{Indyk:2014if}, the proposed algorithm achieves $O(\sparsity \log N )$ sample complexity and $O(N \log^c N)$ computational complexity for general sub-linear sparsity $k$.

In contrast, the 2D-FFAST algorithm is the first to compute an exactly $\sparsity$-sparse $N = N_x\times N_y$-point DFT that has all of the following features:

\begin{itemize}
\item it has $O(\sparsity)$ sample complexity and $O(\sparsity\log\sparsity)$ computational complexity for noiseless observations
\item it can compute a noise robust $\sparsity$-sparse 2D-DFT in sub-linear time of $O(\sparsity\log^{4}N)$ using $O(\sparsity\log^{3}N)$ measurements.
\item it covers the {\em entire} sub-linear regime ($\sparsity \propto N^{\sindex}, 0 < \sindex < 1$)
\item it has a probability of failure that vanishes to zero asymptotically in the problem dimension $N$.
\end{itemize}



\section{Signal model and problem formulation}\label{sec:model}
In this section, we describe our signal model and problem formulation in detail. Specifically, we consider a 2D signal $\msignal \in \complex ^{N_x \times N_y}$ that is a sum of $\sparsity$ complex exponentials:

\begin{equation}\label{eq:2dft}
\msignalc{a}{b} = \sum_{t=0}^{\sparsity-1} \mtransformc{u_t}{v_t}e^{\imath2\pi a u_t/N_x}e^{\imath2\pi b v_t/N_y}
\end{equation}
where the spatial indices $0 \le a \le N_x-1, 0 \le b \le N_y - 1$, the discrete frequencies $ ( u_t, v_t ) \in \{ (0,0), \hdots , (N_x-1, N_y-1) \} $ and the amplitudes $\mtransformc{u_t}{v_t} \in \complex$.

Our signal model assumes that the number of complex exponentials, hence the number of non-zero 2D-DFT coefficients, is sub-linear with respect to the signal dimension, that is $\sparsity \propto N^{\sindex}$, where $N = N_x \times N_y$ and $0 < \sindex < 1$. 

In addition, our 2D-FFAST algorithm requires that the 2D signal dimension $N = N_x \times N_y$ can be factorized into $d$ factors $\{ \pprimef{i} \}_{i=0}^{d-1}$, where $d$ is an appropriately chosen constant (usually chosen as $3$) and each factor is on the order of $k^{1/\sindex d}$ with the factors $\{ \pprimef{i} \}_{i=0}^{d-1}$ being pairwise co-prime. That is,
 
 \begin{equation}
 N = \pprimef{0} \pprimef{1} \hdots \pprimef{d-1}
 \end{equation}
 where $\pprimef{i} \sim k^{1/\sindex d}$ and $\text{gcd}(\pprimef{i},  \pprimef{j}) = 1$ for $i \ne j$. 
 
We note that the dimension assumption is also required by the 1D-FFAST and hence we impose no additional constraint on the 2D signal dimension $(N_x, N_y)$ other than that the total dimension $N= N_x \times N_y$ satisfies the 1D-FFAST dimension requirement. While the dimension assumption limits the application to arbitrary 2D dimensions, this constraint is flexible in practice as long as the user has reasonable control on the signal dimension. For example, suppose we want to target a 2D DFT computation with dimension $N = N_x \times N_y$ around $512 \times 256$ and  sparsity $k  \approx 2500 \approx N^{2/3}$. Then choosing $d = 3$, we have each factor $\pprimef{i} \approx 50$. Hence, we can choose $\pprimef{0} = 49, \pprimef{i} = 50, \pprimef{2} = 51$, and split the factors between $N_x$ and $N_y$ to obtain $N_x = 51 \times 10 = 510$ and $N_y = 5 \times 49 = 245$, which are close to the targeted 2D signal dimension.

For noisy observations, we assume an additive Gaussian noise model, where the observed spatial-domain sampled $y$ are corrupted by white Gaussian noise, that is,

\begin{equation}\label{eq:noisy2dft}
y = x + z
\end{equation}
where $z \in \cnormal ( 0, I )$ ( complex Gaussian i.i.d. noise )

Further, we assume that all the non-zero 2D-DFT coefficients belong to a finite constellation set such that the effective signal-to-noise ratio is finite, that is, $\mathcal {A} = \{ \alpha e^{\imath \Phi} | \alpha \in \mathcal{M}, \phi \in  \Phi \} $, where:
	\begin{itemize}
		\item $\mathcal{M} = \{ \sqrt{ \rho } / 2 + k \sqrt{ \rho } / M_1 \}_{k=0}^{M_1}$ 
		\item $\Phi = \{ 2 \pi k / M_2 \} _{k=0}^{M_2 - 1}$
		\item $\rho$ is a desired signal-to-noise ratio, which is defined as $\rho = \mathbb{E}_{[X[u][v] \ne 0]} \{ |X[u][v]|^2 \} / \mathbb{E} \{ \| z \|^2 / N \}$
		\item $M_1$, $M_2$ are some constants.
	\end{itemize}
The finite constellation assumption is purely for analysis purpose and does not influence the algorithm implementation in practice.

The 2D-FFAST framework proposed in this paper is an extension of the 1D-FFAST architecture from \cite{Pawar:2013da}, hence for ease of exposition, we closely follow the notations used in \cite{Pawar:2013da}, which are summarized in Table~\ref{tab:glossary}.

\def\arraystretch{1.5}
\begin{table}[h]
\begin{center}
\begin{tabular}{ | p{0.15\linewidth} | p{0.8\linewidth} |}
  \hline
  Notation & Description\\
  \hline
  $N_x,N_y$ & Ambient signal size in each of the $2$ dimensions respectively. \\  
  \hline
  $N$ & Ambient signal size $N = N_xN_y$ \\  
  \hline
  $\sparsity$ & Number of non-zero coefficients in the 2D-DFT $\mtransform$.\\ 
  \hline
  $\sindex$& Sparsity-index: $\sparsity \propto N^{\sindex}, 0 < \sindex < 1$.\\
  \hline
  \multirow{2}{*}{$\samples$}& Sample complexity: Number of measurements of $\msignal$ used by the 2D-FFAST \\
 &  algorithm to compute the 2D-DFT $\mtransform$.\\
  \hline
  $\osratio = \samples/k$ & Oversampling ratio: Number of measurements per non-zero DFT coefficient.\\
  \hline
  $\stages$& Number of stages in the ``sub-sampling front end" of the 2D-FFAST architecture.\\
  \hline
\end{tabular}
\end{center}
\caption{Glossary of important notations and definitions used in the rest of the paper.}\label{tab:glossary}
\end{table}

We now describe the Chinese-Remainder-Theorem (CRT) which plays an important role in our proposed 2D-FFAST architecture. For integers $a,n$, we use $(a)_n$ to denote the modulo operation, that is, $(a)_n\triangleq a \mod n$.

\begin{theorem}[Chinese-Remainder-Theorem \cite{Blahut:2010wo}]\label{ch3thm:CRT}
Suppose $n_0,n_1,\hdots,n_{d-1}$ are pairwise co-prime positive integers and $N = \prod_{i=0}^{d-1}n_i$. Then, every integer `$a$' between $0$ and $N-1$ is uniquely represented by the sequence $r_0,r_1,\hdots,r_{d-1}$ of its remainders modulo $n_0,\hdots,n_{d-1}$ respectively and vice-versa. 
\end{theorem}
Further, given a sequence of remainders $r_0,r_1,\hdots,r_{d-1}$, where $0 \leq r_i < n_i$, Gauss's algorithm can be used to find an integer `$a$', such that,
\begin{equation}
(a)_{n_i} \equiv r_i \ \text{ for } i = 0,1,\hdots,d-1.
\end{equation}


\section{Main Result}\label{sec:result}

Given the signal model in Section~\ref{sec:model}, we propose a 2D-FFAST algorithm for computing a $\sparsity$-sparse 2D-DFT with noiseless observations. For noisy observations, we describe a robust extension of the 2D-FFAST called 2D-R-FFAST that aims at robustifying 2D-FFAST that achieves sub-linear computational complexity with (slightly) higher sample complexity for noise robustness. Table~\ref{tab:complex} summarizes the noise model, sample complexity and computational complexity for the proposed algorithms.

\def\arraystretch{1.5}
\begin{table}[h]
\centering
  \begin{tabular}{ | p{0.25\linewidth} || p{0.25\linewidth} | p{0.3\linewidth} |  }
    \hline
    Algorithm & 2D-FFAST & 2D-R-FFAST \\ 
   		   & (See Section~\ref{sec:extension}) & (See Section~\ref{sec:noisy}) \\   \hline
    Noise model  & Noiseless &  Noisy with sufficiently high SNR \\ \hline
    Sample complexity  & $ O( \sparsity )$ &   $ O( \sparsity \log^{3} N  )$ \\ \hline
    Computational complexity & $  O( \sparsity \log \sparsity  )$ &  $ O( \sparsity \log^{4} N )$ \\ \hline
  \end{tabular}
\caption{Summary of sample and computational complexity for proposed algorithms}\label{tab:complex}
\end{table}

Concretely, consider the signal model as stated in Section~\ref{sec:model}, our main results are given by the following theorems:

\begin{theorem}\label{thm:main}
For any $0 < \sindex < 1$, and large enough $N = N_x N_y$, the 2D-FFAST algorithm computes the \sparsity-sparse 2D-DFT  of an $(N_x \times N_y)$-size input $\msignal$, where $\sparsity = O( N^\sindex)$, with the following properties:
\begin{enumerate}
\item {\bf Sample complexity:} The algorithm needs $\samples = O( \sparsity )$ measurements
\item {\bf Computational complexity:}  The computational complexity of the 2D-FFAST algorithm is $O(\samples \log(\samples))$
\item {\bf Probability of success:} The probability that the algorithm correctly computes the \sparsity-sparse 2D-DFT $\mtransform$ \ is at least $1 - O( 1 / m )$ .
\end{enumerate}
\end{theorem}

\begin{proof} For the case when $N_x$ and $N_y$ are co-prime, there exists a unique one-to-one mapping from 2D-DFT to a 1D-DFT as shown in Section~\ref{sec:good}. For the general case, in Appendix~\ref{app:main}, we show that there exists a one-to-one mapping from a 2D-FFAST architecture to a 1D-FFAST architecture.
\end{proof}

\begin{theorem}\label{thm:main_r}
For a sufficiently high signal-to-noise-ratio, any $0 < \sindex < 1$, and a large enough $N = N_x N_y$, the 2D-R-FFAST algorithm computes the \sparsity-sparse 2D-DFT  of an $(N_x \times N_y)$-size input $\msignal$, where $\sparsity = O( N^\sindex)$, with the following properties:
\begin{enumerate}
\item {\bf Sample complexity:} The algorithm needs $\samples = O( \sparsity \log^{3}( N ) )$ measurements
\item {\bf Computational complexity:}  The computational complexity of the 2D-FFAST algorithm is $O( \sparsity \log^{4}( N ) )$
\item {\bf Probability of success:} The probability that the algorithm correctly computes the \sparsity-sparse 2D-DFT $\mtransform$ \ is at least $1 - O( 1 / m )$ .
\end{enumerate}
\end{theorem}

\begin{proof} Please see Appendix~\ref{app:main_r} for the proof for theorem \ref{thm:main_r}.
\end{proof}

Note, that although the main results are for asymptotic values of $\sparsity$, it applies for any signal that has $\zeronorm{\mtransform} \leq \sparsity$. Hence, the regime of $\sindex = 0$ (esp. constant $\sparsity$) is covered by the 2D-FFAST algorithm designed to operate for any $\sindex > 0$, at the expense of being sub-optimal.

\section{FFAST Architecture for 2D Signals with co-prime dimensions }\label{sec:good}
We first discuss a simple extension of the 1D-FFAST \cite{Pawar:2013da} to the 2D setting for the case when the dimensions $N_x$ and $N_y$ are co-prime. When the dimensions are co-prime, there exists a one-to-one mapping between a 2D signal and a 1D signal, such that the 2D-DFT of a 2D signal can be computed using the 1D-DFT of the mapped 1D signal. Hence, one can use the 1D-FFAST architecture proposed in \cite{Pawar:2013da} {\em off-the-shelf}, in conjunction with this mapping, to compute the $\sparsity$-sparse 2D-DFT. 
 
\begin{figure}[h]
\includegraphics[width=\linewidth]{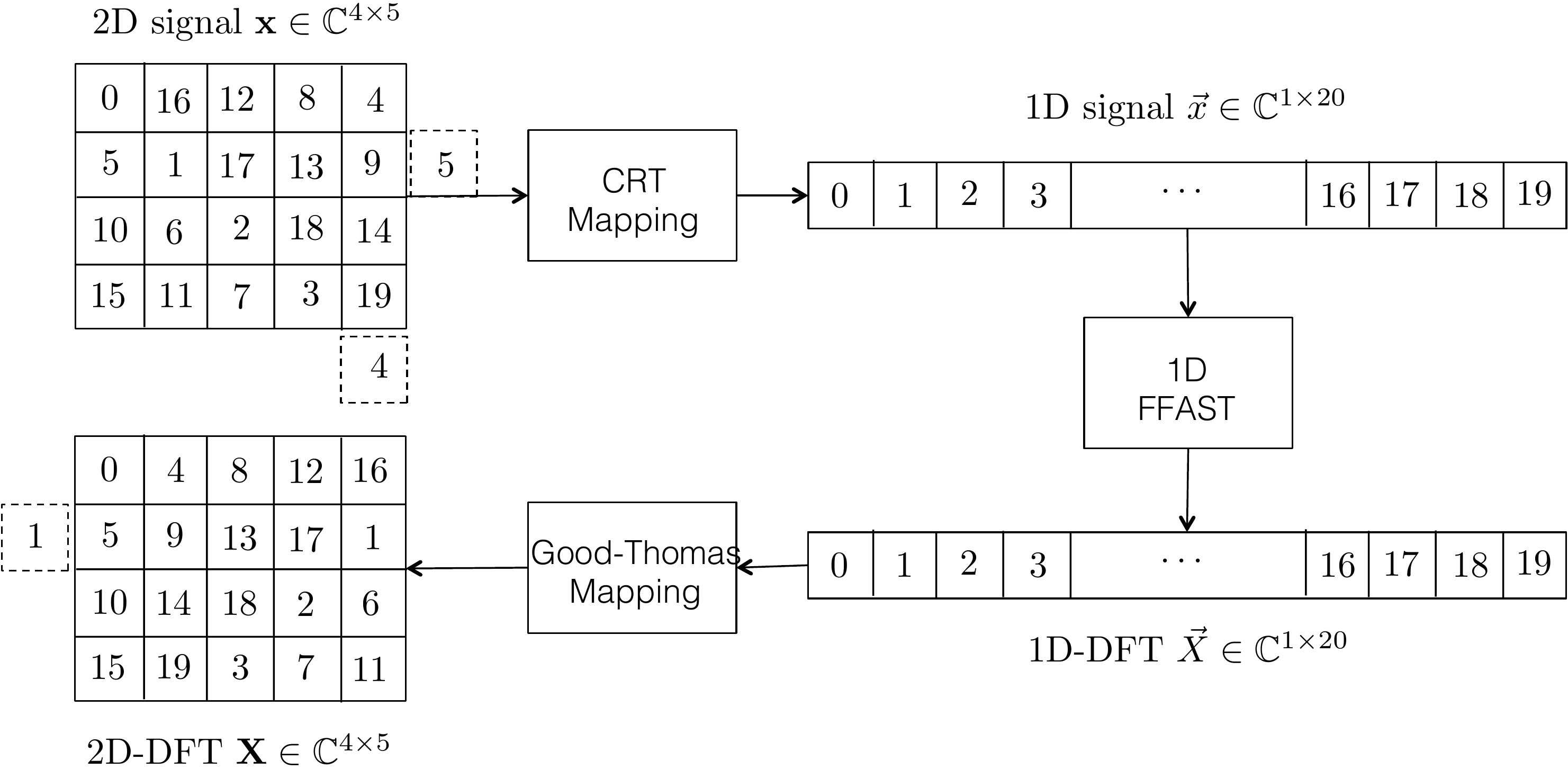}
\caption{Architecture for computing the 2D-DFT of a 2D signal with co-prime dimensions, e.g., $N_x = 4, N_y = 5$. First, we use a forward mapping based on the Chinese-Remainder-Theorem (CRT) to convert a 2D signal $\msignal$ in to a 1D signal $\signal$. Then, we use our 1D-FFAST architecture to compute the 1D-DFT, $\transform$. We transform back the result from 1D-DFT $\transform$, using a reverse mapping based on Good-Thomas algorithm, to get the 2D-DFT $\mtransform$.}
\label{fig:goodmap}
\end{figure}

We will use an example to illustrate the mapping. Consider an example signal $\msignal \in \complex^{4\times5}$, that has a sparse 2D-DFT. We use the Chinese-Remainder-Theorem (CRT) based forward mapping \cite{Good:1958ix,Thomas:1963ux}, to obtain a 1D signal $\signal$ from the 2D signal $\msignal$ as follows: The 1D vector $\signal$ is constructed by reading out the measurements from the 2D signal $\msignal$ in a diagonally downward direction, towards right, starting from top left corner as shown in Fig.~\ref{fig:goodmap}. A uniformly random support in 2D-DFT corresponds to a uniformly random support in CRT mapped 1D-DFT. Hence, we use the 1D-FFAST architecture from \cite{Pawar:2013da}, to compute the sparse 1D-DFT, $\transform$. Then, we perform the reverse mapping of 1D-DFT $\transform$ to the 2D-DFT $\mtransform$ using the Good-Thomas algorithm \cite{Good:1958ix,Thomas:1963ux}. The elements of the 1D vector $\transform$ are read out sequentially and placed in a diagonally downward direction, but towards left, starting from the top-left corner as shown in Fig.~\ref{fig:goodmap}.

In \cite{Good:1958ix,Thomas:1963ux}, the authors show that the forward and the reverse mapping can be done as long as the dimensions are co-prime. Hence, for the signals with co-prime dimensions, the sparse 2D-DFT can be computed using the 1D-FFAST algorithm of \cite{Pawar:2013da}, along with the forward and reverse mapping from \cite{Good:1958ix,Thomas:1963ux}.

\section{FFAST architecture for general 2D signals}\label{sec:extension}

While a sparse 2D-DFT with co-prime dimensions can be computed using the 1D-FFAST, co-prime dimensions can be restrictive and exclude many signals, such as square images. In the following, we describe a 2D-FFAST architecture that can compute a sparse 2D-DFT for a more general class of dimensions as described in Section~\ref{sec:model}.

At a high level, the 2D-FFAST architecture consists of a deterministic sub-sampling ``front-end" and an associated ``back-end" peeling-decoder algorithm as shown in Fig~\ref{fig:2dffast}. Throughout the section, we will use a simple example to illustrate the 2D-FFAST framework.

\begin{figure}[!ht]
\begin{center}
\includegraphics[width=0.8\linewidth]{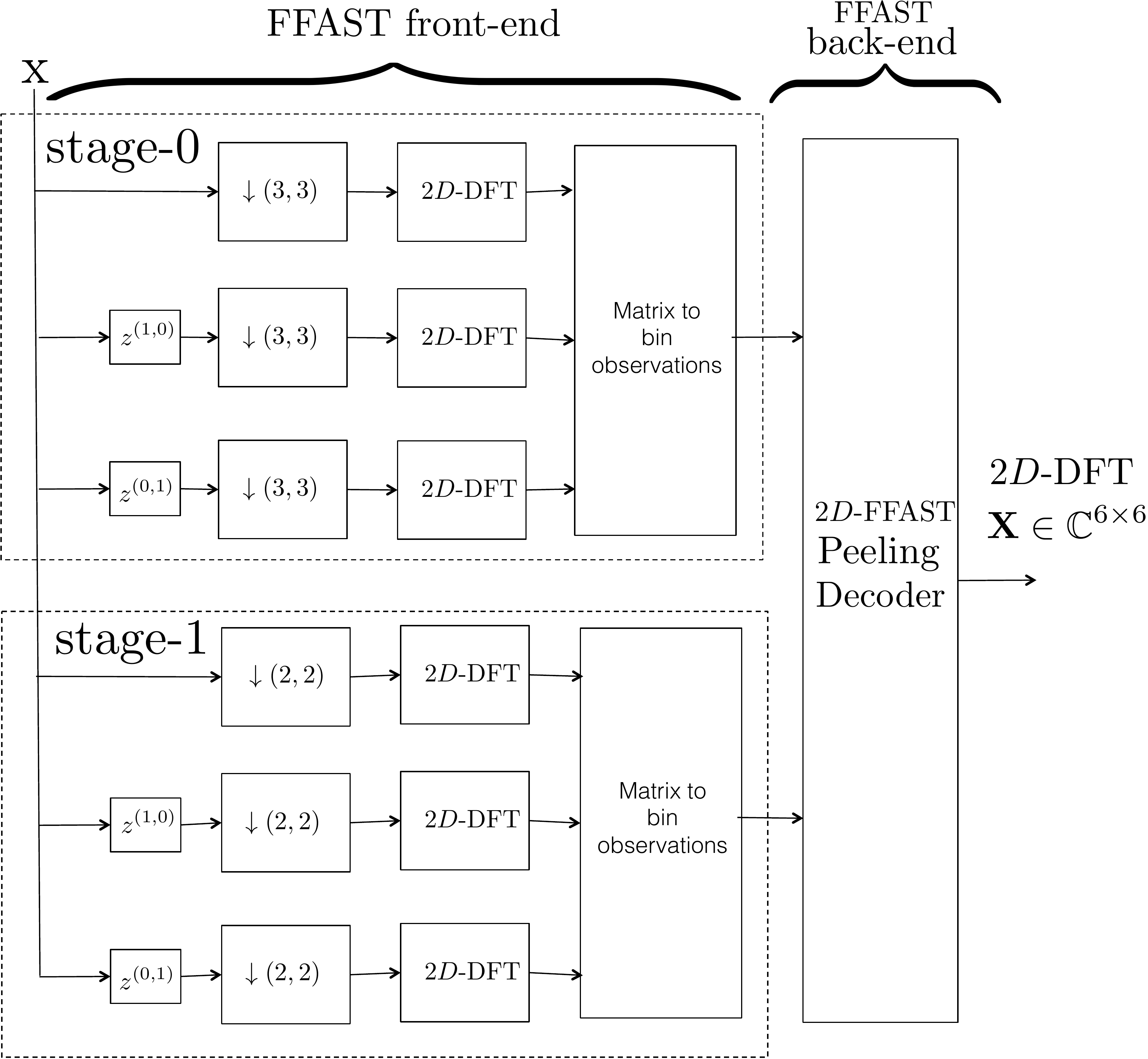}
\caption{A 2D-FFAST architecture of $\stages = 2$ stages. Each stage further has $3$ delay chains and a common sub-sampling factor. The smaller 2D-DFTs are computed of the output of each sub-sampler. Then, bin-observations are formed by collecting one scalar output from each of the $3$ delay chains. The bin-observations are further processed by a `peeling-decoder' to reconstruct the large 2D-DFT $\mtransform$.}
\label{fig:2dffast}
\end{center}
\end{figure}

We choose our example signal to be a $6 \times 6$ 2D signal $\msignal$, such that its 2D-DFT $\mtransform$ is 4-sparse. Let the 4 non-zero 2D-DFT coefficients of $\msignal$ be $\mtransformc{1}{3} = 7, \mtransformc{2}{0} = 3, \mtransformc{2}{3} = 5$ and $\mtransformc{4}{0} = 1$. The 2D-FFAST `front-end' sub-samples the input signal and its circularly shifted version through multiple stages $\stages$. Each stage, further has $\delays = 3$ delay paths and is parametrized by a single sampling factor. For example, the 2D-FFAST architecture of Fig.~\ref{fig:2dffast}, has $\stages = 2$ stages and $3$ delay (circular shift) paths per stage. Let $(n_{x,i}, n_{y,i})$ denote the sampling factor of stage $i$, e.g., in Fig.~\ref{fig:2dffast}, $n_{x,0}=n_{y,0}=3$ and $n_{x,1}=n_{y,1}=2$. Stage $i$ of the 2D-FFAST front-end samples the input signal $\msignal$ and its $3$ circular shifts, $\msignal^{(0,0)}, \msignal^{(1,0)}, \msignal^{(0,1)}$, by $(n_{x,i},n_{y,i})$, in each of the two dimensions.
The 2D-FFAST algorithm synthesizes the big 2D-DFT $\mtransform$, from the short 2D-DFT of each of the sub-sampled stream, using the peeling-decoder. In Section~\ref{sec:binobs}, we describe in detail, the operation of the block {\em matrix-to-bin observations}, which essentially groups the output of the short DFTs in a specific way. 

Before we describe how to compute the 2D-DFT of the signal $\msignal$, using the 2D-FFAST framework, we review some basic signal processing properties of subsampling-aliasing and circular shifts. In Fig.~\ref{fig:2dstage1}, we show a detailed view of stage $0$ of the 2D-FFAST architecture of Fig.~\ref{fig:2dffast}.
\begin{itemize}
\item \textbf{Sub-sampling and aliasing:} If a 2D signal is subsampled in the spatial domain, its 2D-DFT coefficients mix together, that is, alias, in a pattern that depends on the sampling procedure. For example, consider uniform subsampling of $\msignal$ by 3 in both the dimensions. The sub-sampling operation in the first path or delay chain of Fig.~\ref{fig:2dstage1}, results in $\msignal_s =
\left(
\begin{array}{ccc}
  \msignalc{0}{0} &   \msignalc{0}{3}\\
  \msignalc{3}{0}&   \msignalc{3}{3}\\
\end{array}
\right)$. Then, the 2D-DFT coefficients of $\msignal_s$ are related to the 2D-DFT coefficients $\mtransform$ as:

\begin{eqnarray*}
 X_s[0][0] &=& \mtransformc{0}{0} + { \mtransformc{2}{0} } + { \mtransformc{4}{0}} + \mtransformc{0}{2} + \mtransformc{2}{2}\\
 &&+ \ \mtransformc{4}{2} + \mtransformc{0}{4} + \mtransformc{2}{4} + \mtransformc{4}{4} = 4\\
 X_s[0][1] &=& \mtransformc{0}{1} + \mtransformc{2}{1} + \mtransformc{4}{1} + \mtransformc{0}{3} + {\mtransformc{2}{3}}\\
 &&+ \ \mtransformc{4}{3} + \mtransformc{0}{5} + \mtransformc{2}{5} + \mtransformc{4}{5} = 5\\
  X_s[1][0] &=& \mtransformc{1}{0} + \mtransformc{3}{0} + \mtransformc{5}{0} + \mtransformc{1}{2} + \mtransformc{3}{2}\\
 &&+ \ \mtransformc{5}{2} + \mtransformc{1}{4} + \mtransformc{3}{4} + \mtransformc{5}{4} = 0\\
 X_s[1][1] &=&\mtransformc{1}{1} + \mtransformc{3}{1} + \mtransformc{5}{1} + {\mtransformc{1}{3}} + \mtransformc{3}{3}\\
 &&+ \ \mtransformc{5}{3} + \mtransformc{1}{5} + \mtransformc{3}{5} + \mtransformc{5}{5} =  7
\end{eqnarray*}

More generally, if the sampling periods of the row and columns of the 2D-signal $\msignal$ are $n_0$ and $n_1$ (we assume that $n_0$ divides $N_x$ and $n_1$ divides $N_y$) respectively, then,
\begin{equation}
	X_s[i][j] = \sum_{ i\equiv(a)_{N_x / n_0}, j\equiv(b)_{N_y / n_1} } \mtransformc{a}{b}
  \end{equation}
  where notation $i\equiv(a)_{N_x / n_0}$, denotes $i \equiv a$ mod  $(N_x / n_0)$.
  
\item \textbf{Circular spatial shift:} A circular shift in the spatial domain results in a phase shift in the frequency domain. Consider a circularly shifted signal $\msignal^{(1,0)}$ (each column is circularly shifted by $1$) obtained from $\msignal$ as $x^{(1,0)}[i][j] = \msignalc{(i+1)_{N_x}}{j}$. Then the 2D-DFT  coefficients of the shifted signal are given as, $ X^{(1,0)}[i][j] = \omega_{N_x}^i \mtransformc{i}{j}$, where $\omega_{N_x} = \exp( 2\pi \imath  / N_x )$. Similarly, the 2D-DFT  coefficients of a circularly shifted sequence $\msignal^{(0,1)}$ obtained from $\msignal$ as $x^{(0,1)}[i][j] = \msignalc{i}{(j+1)_{N_y}}$, are $ X^{(0,1)}[i][j] = \omega_{N_y}^j \mtransformc{i}{j}$, where $\omega_{N_y} = \exp( 2\pi \imath  / N_y )$. In general, a circular shift of $(s_1,s_2)$ results in $ X^{(s_1,s_2)}[i][j] = \omega_{N_x}^{is_1} \omega_{N_y}^{js_2} \mtransformc{i}{j}$.
\end{itemize}

\begin{figure}[h]
\includegraphics[width=\linewidth]{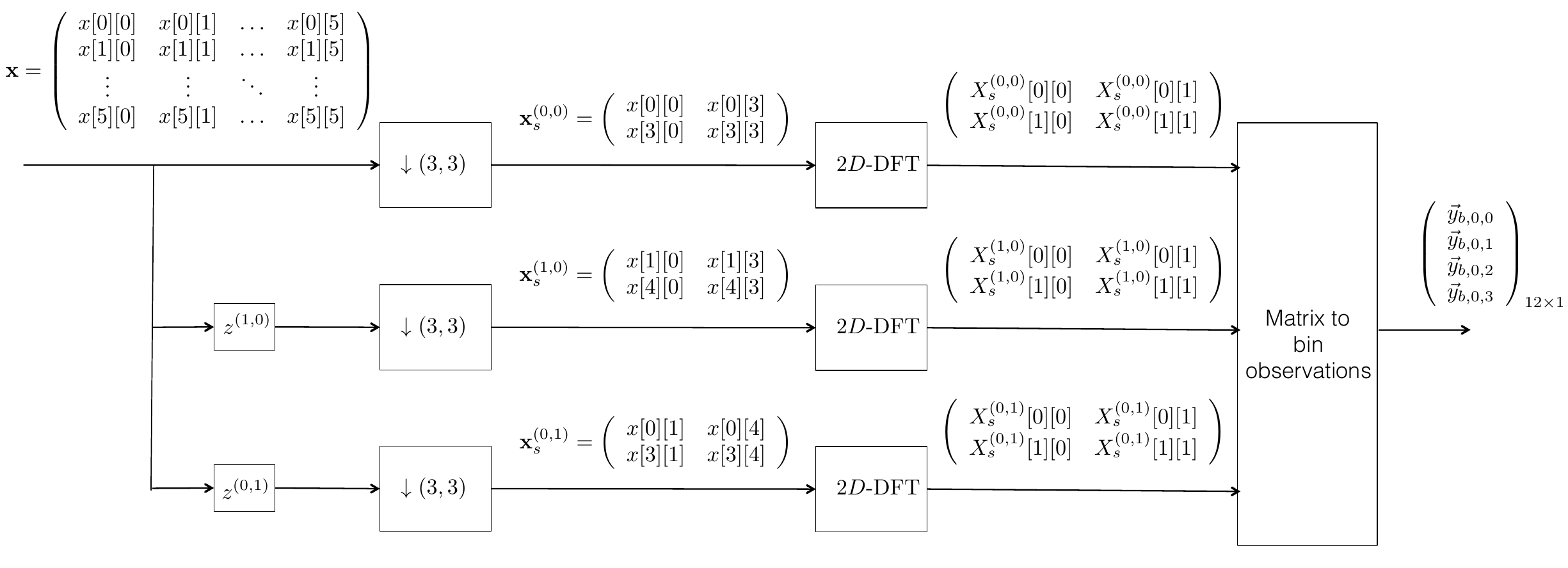}
\caption{A detailed view of stage $0$ of the 2D-FFAST architecture of Fig.~\ref{fig:2dffast}. A 2D signal $\msignal$ and its 2D circularly shifted versions are subsampled by $(3,3)$ to obtain the sub-sampled signals $\msignal^{0,0}_s,\msignal^{1,0}_s$ and $\msignal^{0,1}_s$. Then, the 2D-DFT of the sub-sampled signal is computed and the outputs are collected to form bin-observations.}
\label{fig:2dstage1}
\end{figure}

Using the above signal processing properties of sub-sampling and spatial circular shift, we compute the relation between the 2D-DFT coefficients $\mtransform$ and the output of stage $0$ (shown in Fig.~\ref{fig:2dstage1}) of the 2D-FFAST front-end. Next, we group the output of the 2D-FFAST front end into ``bin-observation" as follows:
\subsubsection{\textbf{Bin observation}}\label{sec:binobs} A bin observation is a $3$-dimensional vector formed by collecting one scalar output value from each of the $3$ delay chains in a stage of the 2D-FFAST front-end. For example, $\binobsv{0}{0}$ is an observation vector of bin $0$ in stage $0$ of the 2D-FFAST front-end and is given by,
\begin{equation}\label{eq:bin0}
\binobsv{0}{0} = 
\left(
\begin{array}{c}
  X^{(0,0)}_{s}[0][0]\\
  X^{(1,0)}_{s}[0][0]\\
  X^{(0,1)}_{s}[0][0]
  \end{array}
\right).
\end{equation}
Note, that there are total of $4$ bins in stage $0$ shown in Fig.~\ref{fig:2dstage1}. The bins are indexed by a single number as follows: a bin formed by collecting the scalar outputs indexed by $(i,j)$ from each of the delay chains is labelled $(N_x/n_0)\times i + j$, e.g., all the delay chain outputs indexed by $(1,0)$ form the observation vector of bin $2$ of stage $0$ (also see Fig.~\ref{fig:sparsegraph1}), and is denoted by
\begin{equation*}
\binobsv{0}{2} = 
\left(
\begin{array}{c}
  X^{(0,0)}_{s}[1][0]\\
  X^{(1,0)}_{s}[1][0]\\
  X^{(0,1)}_{s}[1][0]
  \end{array}
\right).
\end{equation*}

Next, using the above $6\times 6$ example signal $\msignal$, we explain how to compute a sparse 2D-DFT of a signal using decoding over an appropriately designed sparse graph. 

\subsubsection{\textbf{Computing a sparse 2D-DFT via decoding on a sparse-graph}}
\begin{figure}[!ht]
\begin{center}
\includegraphics[width=0.6\linewidth]{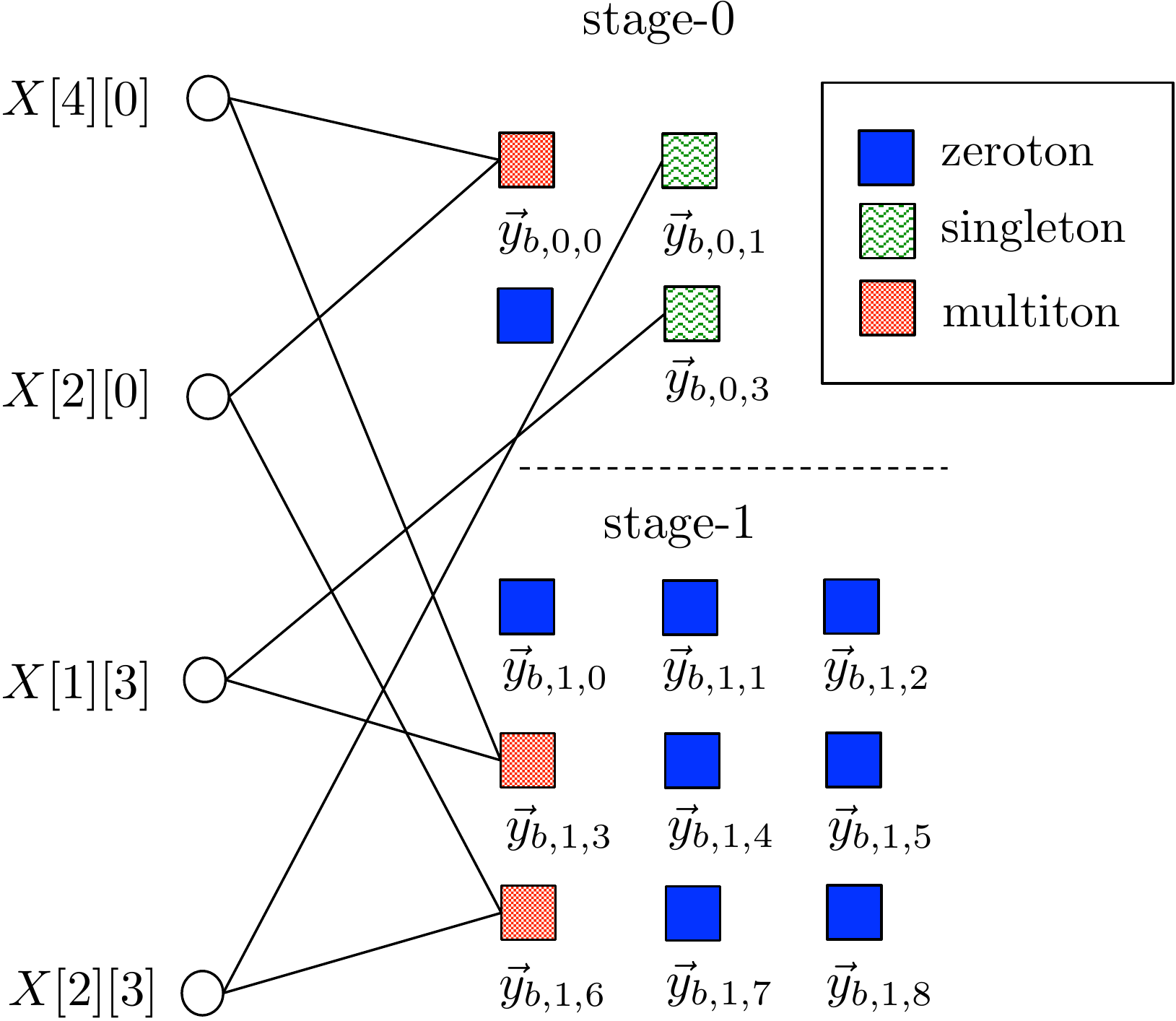}
\caption{A bi-partite graph representing the relation between the bin-observations and the non-zero 2D-DFT coefficients $\mtransform$. Left nodes in the graph represent the non-zero 2D-DFT  coefficients and the right nodes represent the ``bins" (we sometime refer them also as check nodes) with vector observations. An edge connects a left node to a right check node iff the corresponding non-zero 2D-DFT  coefficient contributes to the observation vector at that particular check node, e.g., the 2D-DFT coefficient $\mtransformc{4}{0}$ contributes to the observation vector of bin $0$ of stage $0$ and bin $3$ of stage $1$. A zeroton check node has no left neighbor. A singleton check node has exactly one left neighbor and a multiton check node has more than one left neighbors in the graph.}
\label{fig:sparsegraph1}
\end{center}
\end{figure}

Let the $6\times 6$ example signal $\msignal$ be processed through the $2$-stage 2D-FFAST architecture of Fig.~\ref{fig:2dffast}. Then, the relation between the resulting bin-observations and the non-zero 2D-DFT coefficients $\mtransform$ can be computed using the sub-sampling and the circular shift properties explained earlier and is graphically represented in Fig.~\ref{fig:sparsegraph1}. Left nodes of the graph in Fig.~\ref{fig:sparsegraph1} represent the non-zero DFT coefficients and the right nodes represent the ``bins" (check nodes) with vector observations. An edge connects a left node to a right check node iff the corresponding non-zero 2D-DFT  coefficient contributes to the observation vector of that particular check node, e.g., after aliasing, due to sub-sampling, the 2D-DFT coefficient $\mtransformc{4}{0}$ contributes to the observation vector of bin $0$ of stage $0$ and bin $3$ of stage $1$.

We define the following:
\begin{itemize}
\item {\bf zero-ton}: A bin that has no contribution from any of the non-zero DFT coefficients of the signal, e.g., bin $2$ of stage $0$ or bin $0$ of stage $1$ in Fig.~\ref{fig:sparsegraph1}. A zero-ton bin can be trivially identified from its observations.
\item {\bf single-ton}: A bin that has contribution from exactly one non-zero DFT coefficient of the signal, e.g., bin $1$ of stage $0$. Using the signal processing properties the observation vector of bin $1$ of stage $0$ is given as,
\begin{equation*}
\binobsv{0}{1} = 
\left(
\begin{array}{c}
  X[2][3]\\
  e^{2\pi\imath2/6}X[2][3]\\
  e^{2\pi\imath3/6}X[2][3]
  \end{array}
\right).
\end{equation*}
The observation of a singleton bin can be used to determine the 2D support and the value, of the only non-zero DFT coefficient contributing to that bin, as follows:
\begin{itemize}
\item {\em Row support}: The row-support of the non-zero DFT coefficient contributing to a singleton bin can be computed as,
\begin{equation}
2 = \frac{6}{2\pi}\angle y_{b,0,1}[1]y^{\dagger}_{b,0,1}[0]
\end{equation}
\item {\em Column support}: The column-support of the non-zero DFT coefficient contributing to a singleton bin can be computed as,
\begin{equation}
3 = \frac{6}{2\pi}\angle y_{b,0,1}[2]y^{\dagger}_{b,0,1}[0]
\end{equation}
\item {\em Value}: The value of the non-zero DFT coefficient is given by the observation $y_{b,0,1}[0]$.
\end{itemize}
We refer to this procedure as a ``ratio-test", hence forth. Thus, a simple ratio-test on the observations of a singleton bin correctly identifies the 2D support and the value of the only non-zero DFT coefficient connected to that bin. This property holds for all the singleton bins.

\item {\bf multi-ton}: A bin that has contribution from more than one non-zero DFT coefficients of the signal, e.g., bin $0$ of stage $0$. The observation vector of bin $0$ of stage $0$ is,
\begin{equation*}
\binobsv{0}{0} = 
\left(
\begin{array}{c}
  X[4][0] + X[2][0]\\
  e^{2\pi\imath4/6}X[4][0] +   e^{2\pi\imath2/6}X[2][0]\\
 X[4][0] + X[2][0]
  \end{array}
\right) = \left(
\begin{array}{c}
  4\\
  -2 + \imath \sqrt{3}\\
 4
  \end{array}
\right).
\end{equation*}
Now if we perform the ``ratio-test" on these observations, we get the column support to be $1.3484$. Since we know that the column support has to be an integer value between $0$ to $5$, we conclude that the observations do not correspond to a singleton bin. In \cite{Pawar:2013da}, we rigorously show that the ratio-test identifies a multi ton bin almost surely. 
\end{itemize}

Hence, using the ``ratio-test" on the bin-observations, the 2D-FFAST decoder can determine if a bin is a zero-ton, a single-ton or a multi-ton, almost surely. Also, when the bin is singleton the ratio-test provides the 2D-support as well as the value of the non-zero DFT coefficient connected to that bin. We use the following peeling-decoder on the graph in Fig.~\ref{fig:sparsegraph1}, to compute the support and the values of the non-zero DFT coefficients of $\msignal$.

{\bf Peeling-Decoder:} The peeling-decoder repeats the following steps:
\begin{enumerate}
\item Select all the edges in the graph with right degree $1$ (edges connected to singleton bins).
\item Remove these edges from the graph as well as the associated left and the right nodes.
\item Remove all other edges that were connected to the left nodes removed in step-2. When a neighboring edge of any right node is removed, its contribution is subtracted from that bins observation.
\end{enumerate}
Decoding is successful if, at the end, all the edges have been removed from the graph. It is easy to verify that performing the peeling procedure on the graph of Fig.~\ref{fig:sparsegraph1}, results in successful decoding with the coefficients being uncovered in the following possible order: $\mtransformc{2}{3},  \mtransformc{1}{3}, \mtransformc{4}{0}, \mtransformc{2}{0}$.

Thus, the 2D-FFAST architecture of Fig.~\ref{fig:2dffast}, has transformed the problem of computing a 2D-DFT of a signal $\msignal$ into that of decoding over a sparse graph of Fig.~\ref{fig:sparsegraph1}, induced by the sub-sampling operations. The success of the peeling-decoder depends on the properties of the induced graph. In Appendix~\ref{app:main} we show that, if the sub-sampling parameters are chosen carefully, guided by the Chinese-Remainder-Theorem (CRT), the 2D-FFAST decoder succeeds with probability approaching $1$, asymptotically in $\sparsity$.

\section{Noise Robust 2D-FFAST Algorithm }\label{sec:noisy}

In this section, we provide a brief overview of robust extensions of 2D-FFAST. The robustness against sample noise is achieved by a simple modification to the 2D-FFAST framework. Specifically, the 2D-FFAST front-end in Section~\ref{sec:extension} has multiple stages of sub-sampling operations, where each stage further has 3 delay-chains to estimate the location of the singletons. For noiseless observations, the location estimation for each dimension is done by two separable ratio-tests. However, for noisy observations, the ratio-test becomes unreliable and more measurements are needed for a robust support estimation. Hence, to become robust to sample noise, we modify the front end by increasing the number of delay-chains and modify the back-end by replacing the `ratio-test' with a more sophisticated singleton estimator.

In~\cite{Pawar:2014gx}, we have shown that a noiseless 1D-FFAST framework can be made {\em noise-robust} by using $O(\log^{3} N)$ number of delay-chains per sub-sampling stage, achieving sub-linear computational complexity. Since for 2D-DFT, we follow a {\em separable} approach and perform $2$ independent ratio-tests to determine the 2D support of a single-ton bin, all the results from \cite{Pawar:2014gx} follow. With $O(\log^{3} N)$ specially designed delays, the 2D-R-FFAST algorithm computes a $\sparsity$-sparse 2D-DFT in sub-linear time of $O(\sparsity\log^{4}N)$ using $O(\sparsity\log^{3}N)$ measurements.

\section{Simulations}\label{sec:simulations}
In this section we validate the empirical performance of the 2D-FFAST architecture for a range of settings gradually varying from synthetic signals, generated using the signal model specified in the paper, to real world MR images. In Section \ref{sec:sim_exact}, we provide simulations for synthetic 2D signals with an exactly sparse 2D-DFT and compare it with the theoretical claims of this paper. In Section \ref{sec:sim_clustered}, we evaluate our 2D-FFAST algorithm on synthetically generated `Cal image' and `house image' (see Fig. \ref{fig:clustered}) that although are exactly sparse, do not have uniformly random support for the non-zero DFT coefficients. Later in Section~VII-C, we showcase the noise robustness of our algorithm using the noisy "cal image" as an input to our algorithm. Lastly, in Section~VII-D, we show an application of the 2D-FFAST architecture to acquire Magnitude Resonance Imaging (MRI) data, which is not exactly sparse and also has a structured, rather than a uniformly random, support for the dominant 2D-DFT coefficients. Thus, providing an empirical evidence that the 2D-FFAST architecture is applicable to realistic signals, although the theoretical claims are only for signals with an exactly sparse 2D-DFT with uniformly random support.

In the spirit of reproducible research, we provide a software package to reproduce most of the results described in this paper. The software package can be downloaded from:
\begin{center}
\url{https://www.eecs.berkeley.edu/~kannanr/project_ffft.html}
\end{center}

\subsection{Application of 2D-FFAST for signals with exactly $\sparsity$-sparse 2D-DFT}\label{sec:sim_exact}

\begin{figure}[!ht]
\begin{center}
\includegraphics[width=\linewidth]{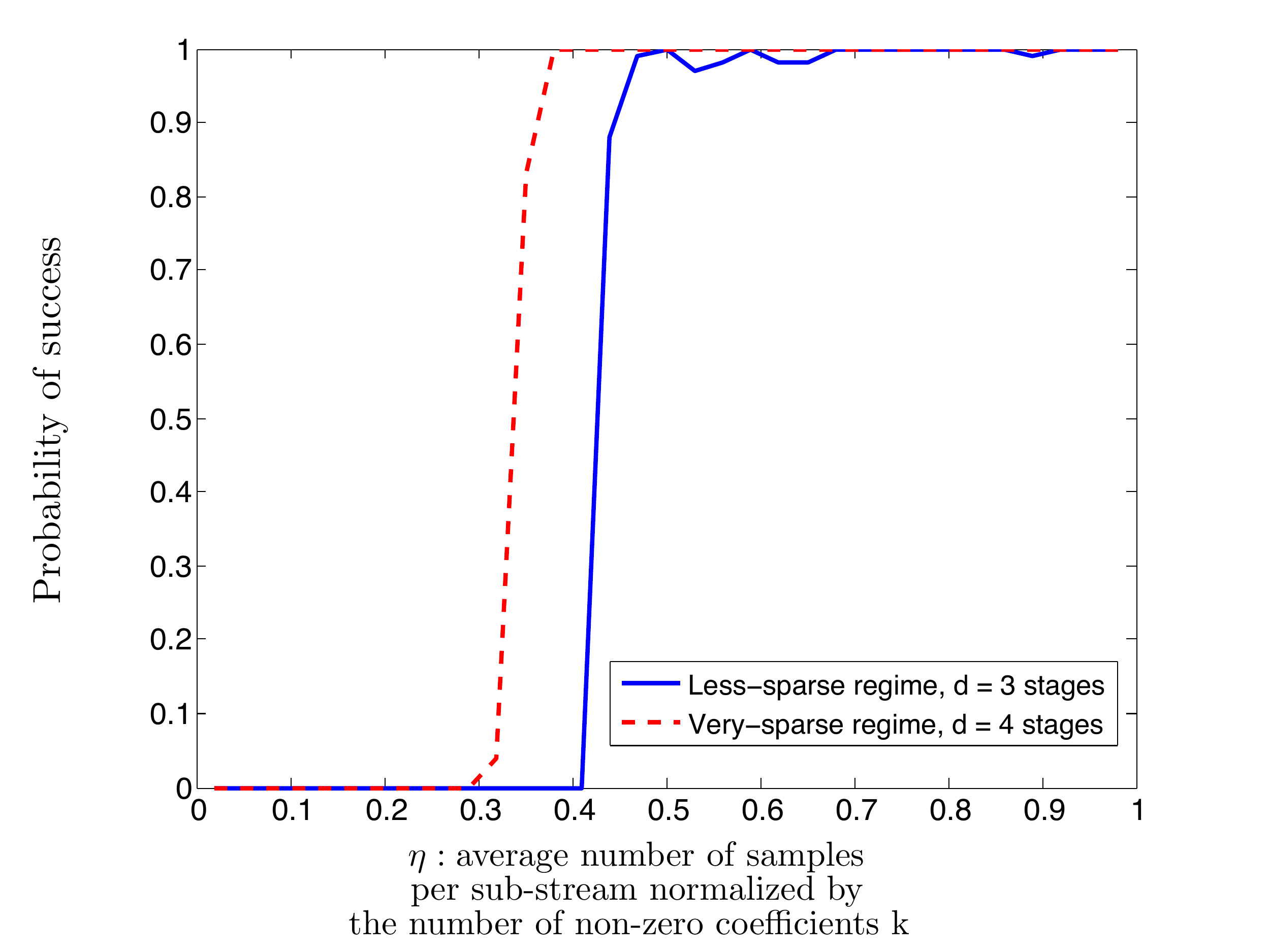}
\caption{The probability of success of the 2D-FFAST algorithm as a function of $\eta$, the average number of bins per stage normalized by the number of non-zero coefficients $\sparsity$. The plot is obtained for the two different sparsity regimes: 1) Very-sparse regime, that is, $\sindex \leq 1/3$. For this regime, a $\stages = 4$ stage 2D-FFAST architecture is used; 2) Less-sparse regime, that is, $1/3 < \sindex < 1$. For this regime, a $\stages = 3$ stage 2D-FFAST architecture is used. Each point in the plot is obtained by averaging over 100 runs of the simulations. The ambient signal dimension $N_x\times N_y$, and the number of measurements $\samples$ are fixed in both the simulations, while the number of the non-zero 2D-DFT coefficients $\sparsity$, is varied to get different values of $\eta$. We note that the 2D-FFAST algorithm exhibits a threshold behavior with the thresholds $\eta_1 = 0.38$ and $\eta_2 = 0.47$ for the very-sparse and the less-sparse regimes respectively. From \cite{Pawar:2013da}, we know that the optimal thresholds for a $\stages = 4$ and $\stages = 3$ stage 1D-FFAST architecture are $\eta^*_1 = 0.3237$ and $\eta^*_2 = 0.4073$ respectively. Thus, confirming that the empirical behavior of the 2D-FFAST architecture is in close agreement with the theoretical results.}
\label{fig:ksparse}
\end{center}
\end{figure}

Figure \ref{fig:ksparse} shows a probability of success plot (phase transition curve) generated from simulations. The plot demonstrates the number of samples required to perfectly reconstructed an exactly $\sparsity$-sparse 2D-DFT in a particular setting and is in close agreement with the theoretical results as described in \cite{Pawar:2013da}. The details of the simulation are as follows:

\begin{itemize}

\item \textbf{Very sparse regime $0 < \sindex \leq 1/3$:} We construct a 2D signal $\msignal$ of ambient dimensions $N_x = N_y = 2520$. The sparsity parameter $\sparsity$ is varied from $70$ to $170$ which corresponds to the very-sparse regime of $\sparsity < O(N^{1/3})$, where $N = N_xN_y = 6.35$ million. We use a $4$ stage 2D-FFAST architecture (see Fig.~\ref{fig:2dffast} for reference). The uniform sampling periods for each of the 4 stages are $(5 \times 7 \times 8)^2$, $(7 \times 8 \times 9)^2$, $(8 \times 9 \times 5)^2$ and $(9 \times 5 \times 7)^2$ respectively. Thus the number of bins/check-nodes in the $4$ stages are $9^2$, $5^2$, $7^2$ and $8^2$ respectively. Further, each stage has $3$ delay sub-streams. Thus, the total number of measurements used by the 2D-FFAST algorithm for this simulation is $\samples < 3(9^2 + 5^2 + 7^2+ 8^2) = 657$. We define $\eta$ as an average number of bins per stage normalized by the sparsity $\sparsity$. For example, when $\sparsity = 100$ the value of $\eta = 0.5475$. In \cite{Pawar:2013da}, we have shown that for a $4$-stage 1D-FFAST architecture the theoretical threshold $\eta^* = 0.3237$. We observe that the 2D-FFAST algorithm also shows a threshold behavior and successfully reconstructs the 2D-DFT for all values of $\eta > 0.38$, which is in close agreement with the theoretical claims.

\item \textbf{Less sparse regime $1/3 < \sindex < 1$:} We construct a 2D signal $\msignal$ of ambient dimensions $N_x = N_y = 280$. The sparsity parameter $\sparsity$ is varied from $2000$ to $5000$ which corresponds to the less-sparse regime, $O(N^{2/3}) < \sparsity < O(N^{3/4})$, where $N = N_xN_y = 78400$. We use a $\stages = 3$ stage 2D-FFAST architecture (see Fig.~\ref{fig:2dffast} for reference). The uniform sampling periods for each of the 3 stages are $5,8$ and $7$ respectively. Thus the number of bins/check-nodes in the $3$ stages are $35^2$, $56^2$ and $40^2$ respectively. Further, each stage has $3$ delay sub-streams. Thus the total number of measurements used by the 2D-FFAST algorithm for this simulation is $\samples < 3(35^2 + 56^2 + 40^2) = 17883$. From \cite{Pawar:2013da}, for a $3$ stage 1D-FFAST architecture, the theoretical threshold $\eta^* = 0.4073$. We observe that the 2D-FFAST algorithm successfully reconstructs the 2D-DFT for all values of $\eta > 0.47$, which is in close agreement with the theoretical thresholds.

\item For each run of the simulation, a 2D-DFT $\mtransform$ of dimension $N_x \times N_y$ is generated, with exactly $\sparsity$ non-zero coefficients. The support of the non-zero DFT coefficients is chosen uniformly at random from the set $ \left\{0,1,...,N_x-1 \right\} \times \left\{ 0,1,...,N_y-1 \right\}$. The spatial 2D signal $\msignal$ is then generated from $\mtransform$ using an inverse 2D-DFT operation. This discrete signal $\msignal$ is then given as an input to the 2D-FFAST front-end.

\item Each sample point in Fig.~\ref{fig:ksparse} is generated by averaging over 100 runs of the simulations.

\item Decoding is successful if all the 2D-DFT coefficients are recovered perfectly.

\end{itemize}

\begin{figure}[!ht]
\begin{center}
\includegraphics[width=\linewidth]{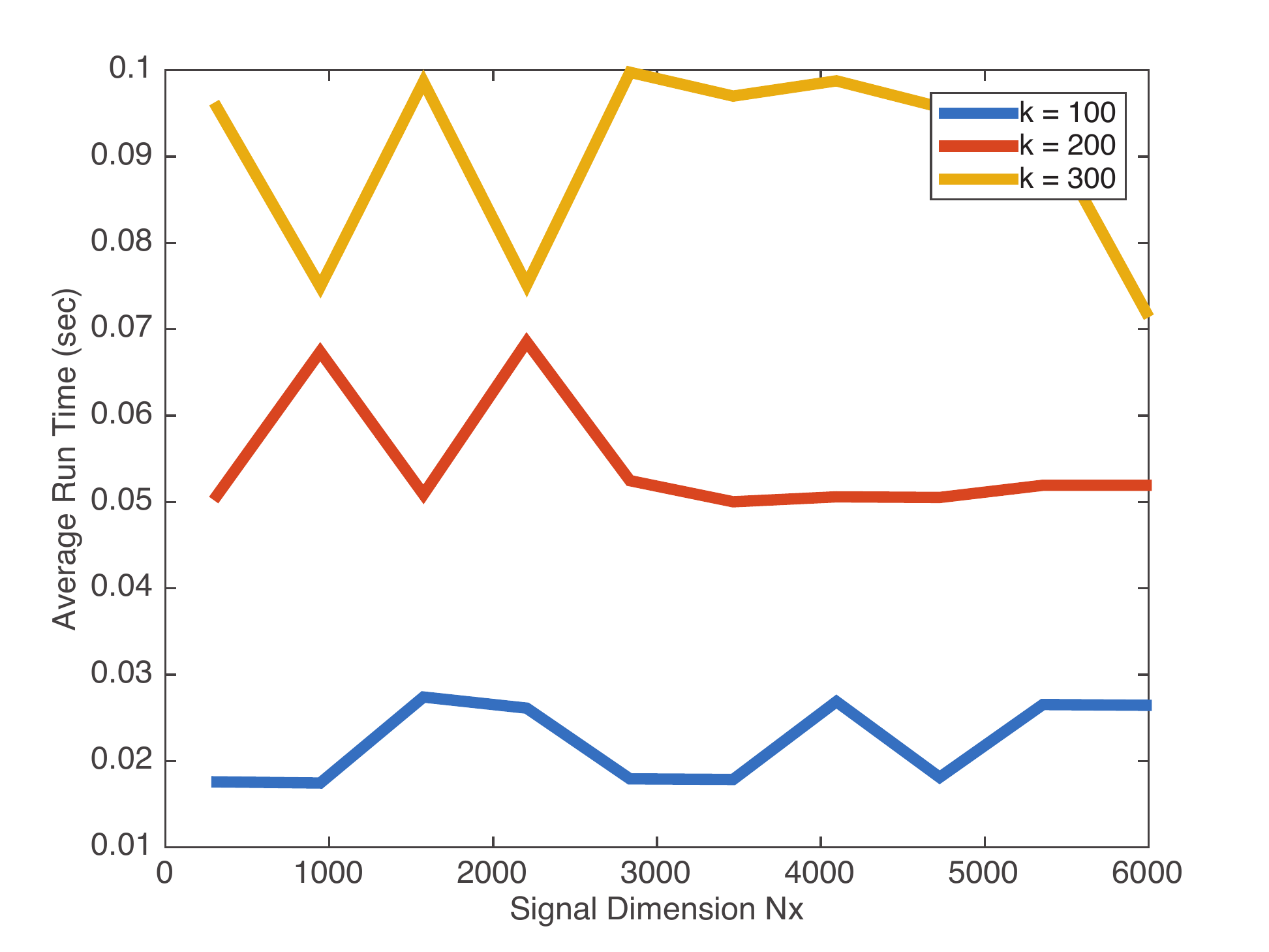}
\caption{Empirical run-time of the 2D-FFAST for different sparsity level $k$ with varied $N_x$ and fixed $N_y$}
\label{fig:time}
\end{center}
\end{figure}

Figure \ref{fig:time} shows the empirical run-time of the 2D-FFAST with $\sparsity = 100, 200, 300$ and ambient dimension $N_x$ varied from $315$ to $5985$ with $N_y$ held fixed at $315$. We used a $\stages = 3$ stage 2D-FFAST architecture with $3$ delays per stage. Each problem instance was run $1000$ times, and the averaged run-time was plotted. As seen from the figure, the empirical run-time remains constant irrespective of the growth in the ambient dimension $Nx$ and only grows with the sparsity $\sparsity$, thus confirming our theoretical claim of computational complexity being independent of ambient dimension N.

\subsection{Application of 2D-FFAST for signals with exactly $\sparsity$-sparse 2D-DFT but with non-uniform support}\label{sec:sim_clustered}
\begin{figure}[!ht]
\centering
\subfigure[Log-intensity plot of the 2D-DFT of the `Cal' image of dimension $280 \times 280$.]{\label{fig:cal_dft}
\includegraphics[width=.31\linewidth]{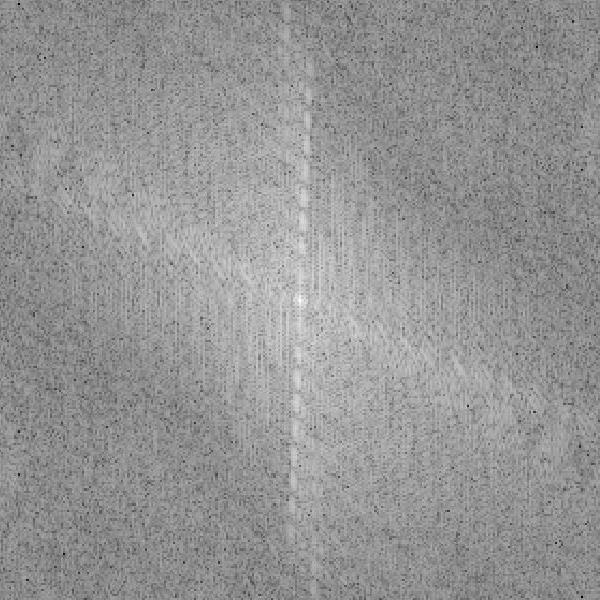}}
\subfigure[2D-FFAST subsampled 2D-DFT of the `Cal' image with $\samples=4.75k=16668$ measurements. The white pixels correspond to the sampled data.]
{\includegraphics[width=.31\linewidth]{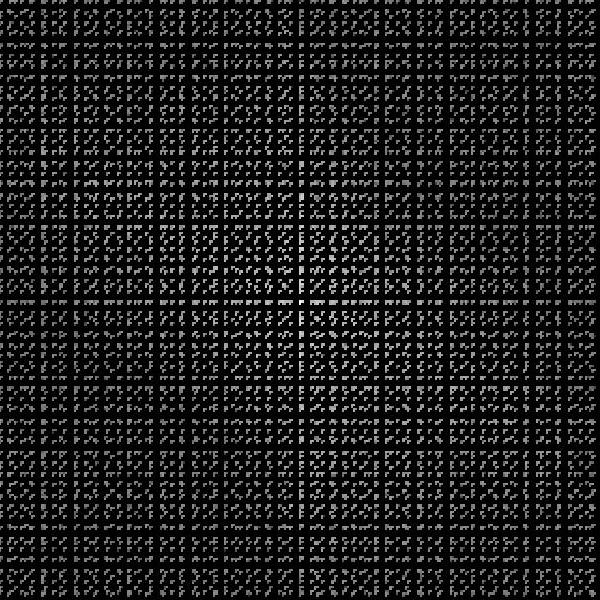}}
\subfigure[Perfectly reconstructed `Cal' image using the 2D-FFAST algorithm]{\label{fig:cal_spatial}\includegraphics[width=.31\linewidth]{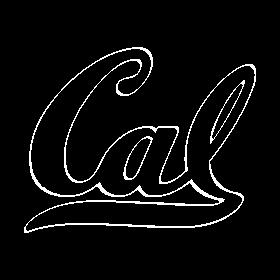}}

\subfigure[Log-intensity plot of the 2D-DFT of the `House' image of dimension$247 \times 238$.]{\label{fig:house_a}\includegraphics[width=.31\linewidth]{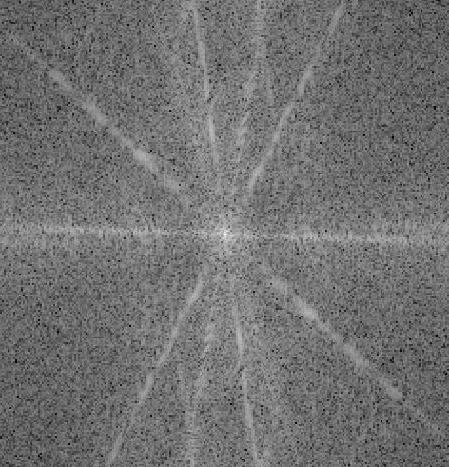}}
\subfigure[1D-FFAST subsampled 2D-DFT of `House' image with $\samples=5.46k=25126$ measurements.]{\label{fig:house_b}\includegraphics[width=.31\linewidth]{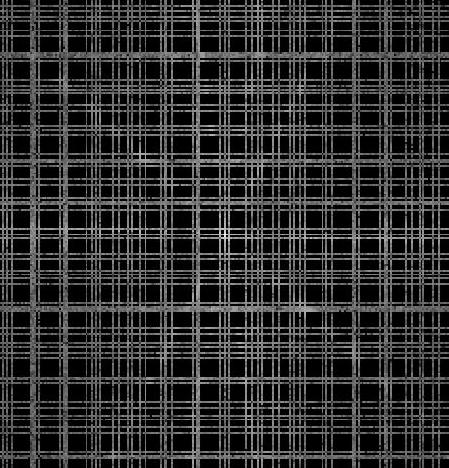}}
\subfigure[Perfectly reconstructed `House' image using prime-factor mapping and the 1D-FFAST algorithm]{\label{fig:house_c}\includegraphics[width=.31\linewidth]{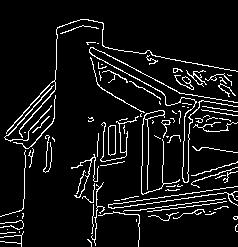}}
\caption{The figure shows the performance of the 2D and 1D FFAST algorithm for signals with exactly-sparse spatial image domain, with non-uniform support of the non-zero pixels. The signal is sampled in the Fourier domain and the 2D-FFAST algorithm is used to perfectly recover the image domain data.}\label{fig:clustered}
\end{figure}
The theoretical guarantees of Theorem~\ref{thm:main} are for signals that have an exactly $\sparsity$-sparse 2D-DFT, with the support of the non-zero DFT coefficients being uniformly random. In this section we empirically show that the proposed 2D-FFAST architecture works quite well even for signals that have non-uniform support of the 2D-DFT coefficients. We provide simulation results for two different types of images with the parameters specified as below.

\begin{itemize}
\item \textbf{ `Cal' Image: } The `Cal' image shown in Fig.~\ref{fig:cal_spatial}, is a synthetic image of size $280\times 280$. The number of the non-zero pixels in the `Cal' image is $\sparsity = 3509$, and the support is non-uniform. Note, since the image is sparse in the spatial domain, we sub-sample in the frequency domain, that is, the input to the 2D-FFAST front-end is the 2D-DFT of the `Cal' image (see Fig.~\ref{fig:cal_dft}). We use a $3$ stage 2D-FFAST architecture with sampling periods $5,7$ and $8$ in each of the $3$ stages respectively. The 2D-FFAST algorithm perfectly reconstructs the spatial image using $\samples = 16668$ measurements of its 2D-DFT domain.

\item \textbf{ `House' Image: } The `House' image shown in Fig.~\ref{fig:house_c} is created by applying the Canny edge detection algorithm on the commonly used House image in the image processing literature. The `House' image is cropped to size $247 \times 238$. The number of non-zero pixels is $\sparsity = 4599$ and the support is non-uniform. Since the dimensions of the image are co-prime, we use the prime factor mapping of Section~\ref{sec:good}, and apply the 1D-FFAST algorithm to perfectly reconstruct the image using $\samples = 5.46 \sparsity = 25126$ measurements of its 2D-DFT domain.

\end{itemize}

\subsection{Application of 2D-R-FFAST for signals with exactly $\sparsity$-sparse 2D-DFT and additive noise}\label{sec:sim_noisy}

\begin{figure}[!ht]
\centering
\subfigure[Noisy `Cal' image with SNR = $13$dB]
{\label{fig:caln_spect}
\includegraphics[width=.31\linewidth]{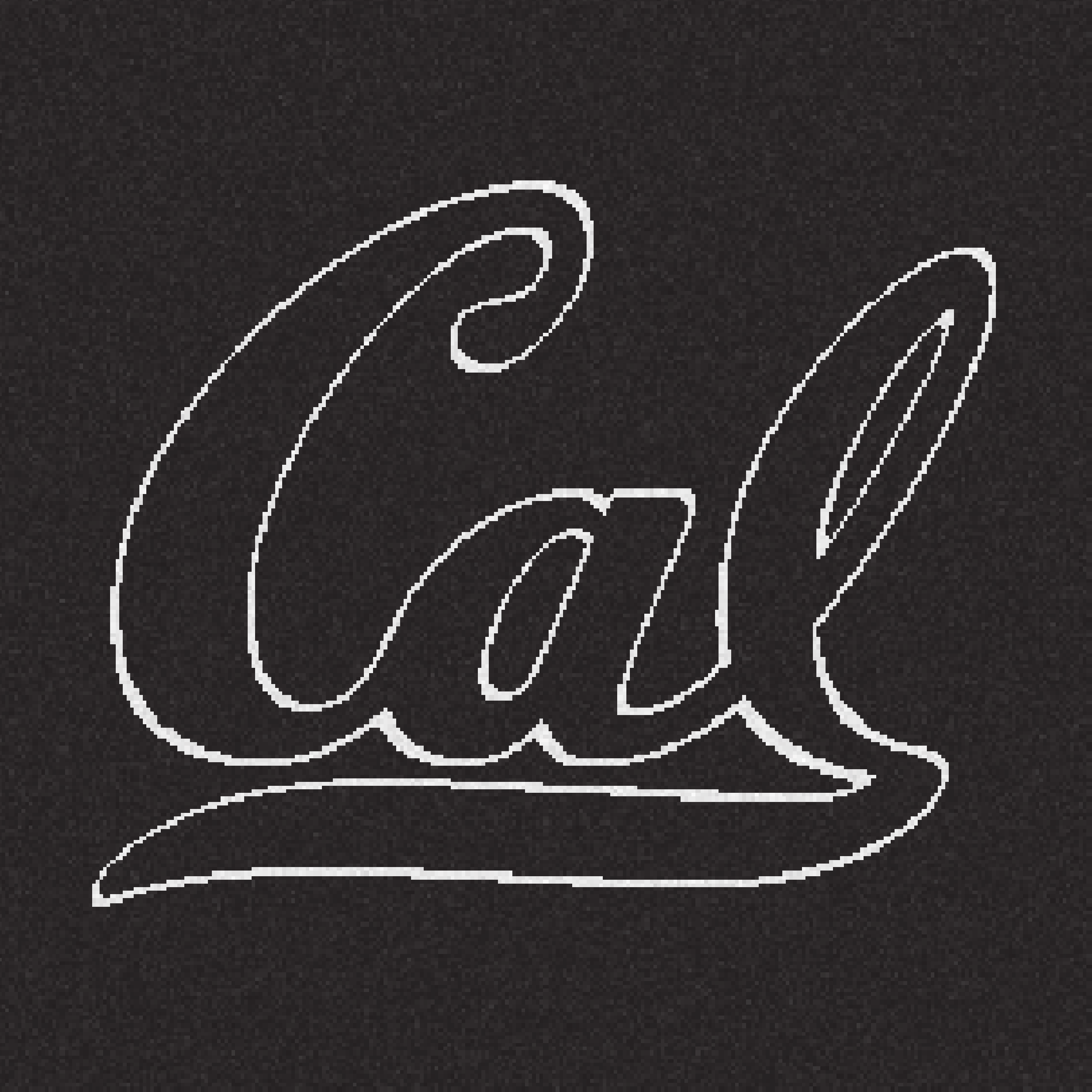}}
\subfigure[2D-FFAST subsampled 2D-DFT of the noisy `Cal' image with $\samples=12.3k=43200$ measurements. ]
{\label{fig:caln_signal}
\includegraphics[width=.31\linewidth]{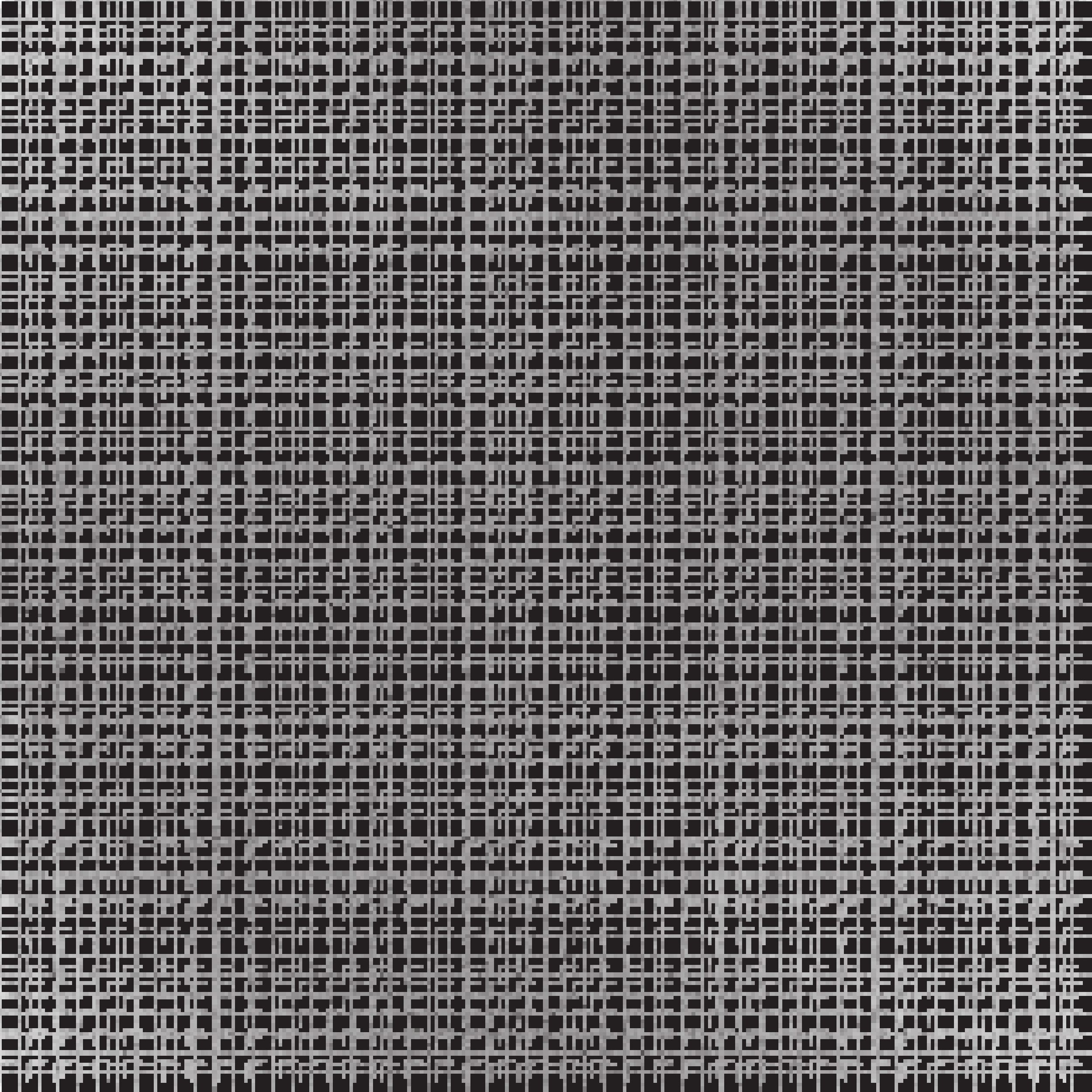}}
\subfigure[Reconstructed `Cal' image using the 2D-FFAST algorithm]
{\label{fig:caln_spect_rec}
\includegraphics[width=.31\linewidth]{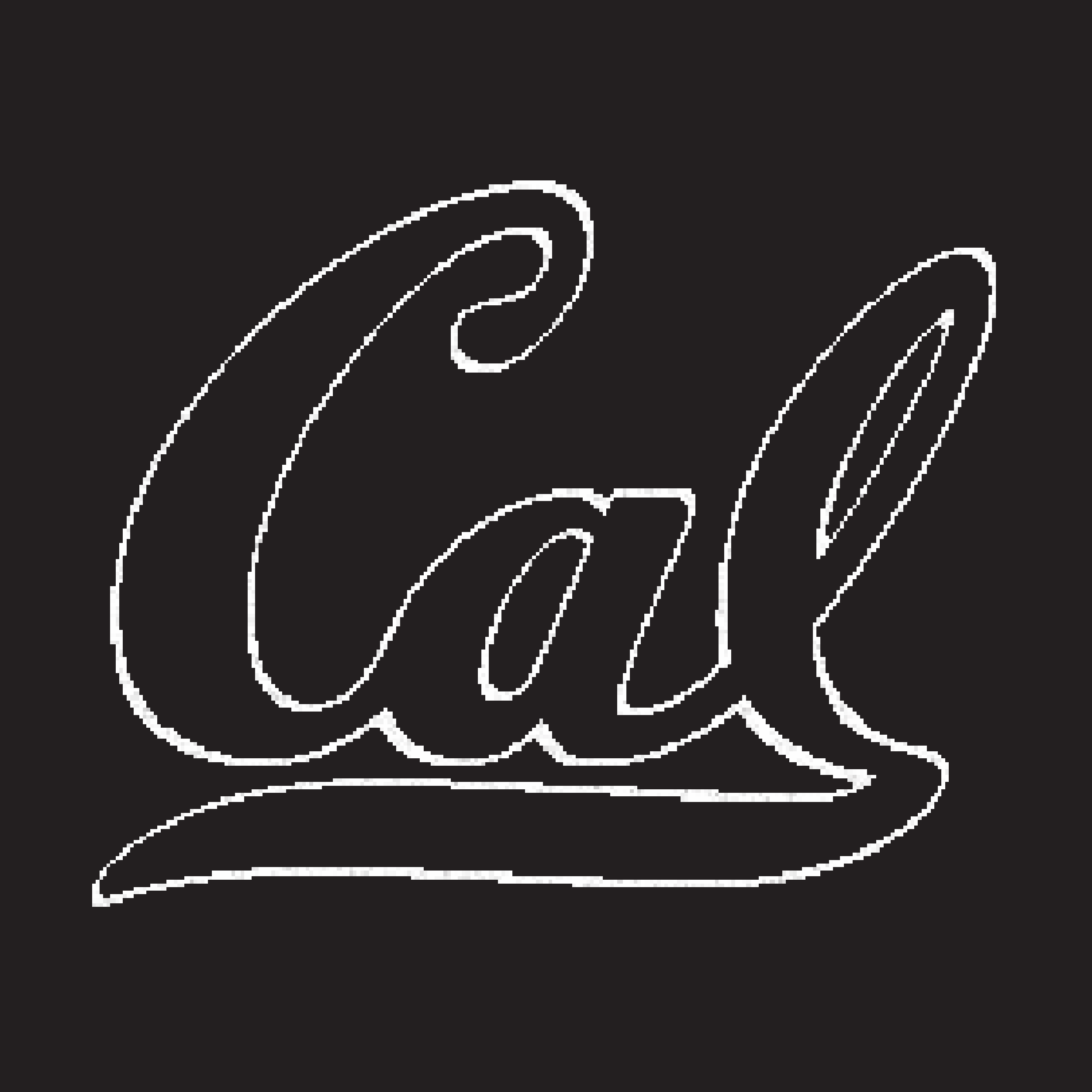}}

\subfigure[Difference between the noisy and original `Cal' images. Amplified by 5 for visualization]
{\label{fig:caln_diff}
\includegraphics[width=.31\linewidth]{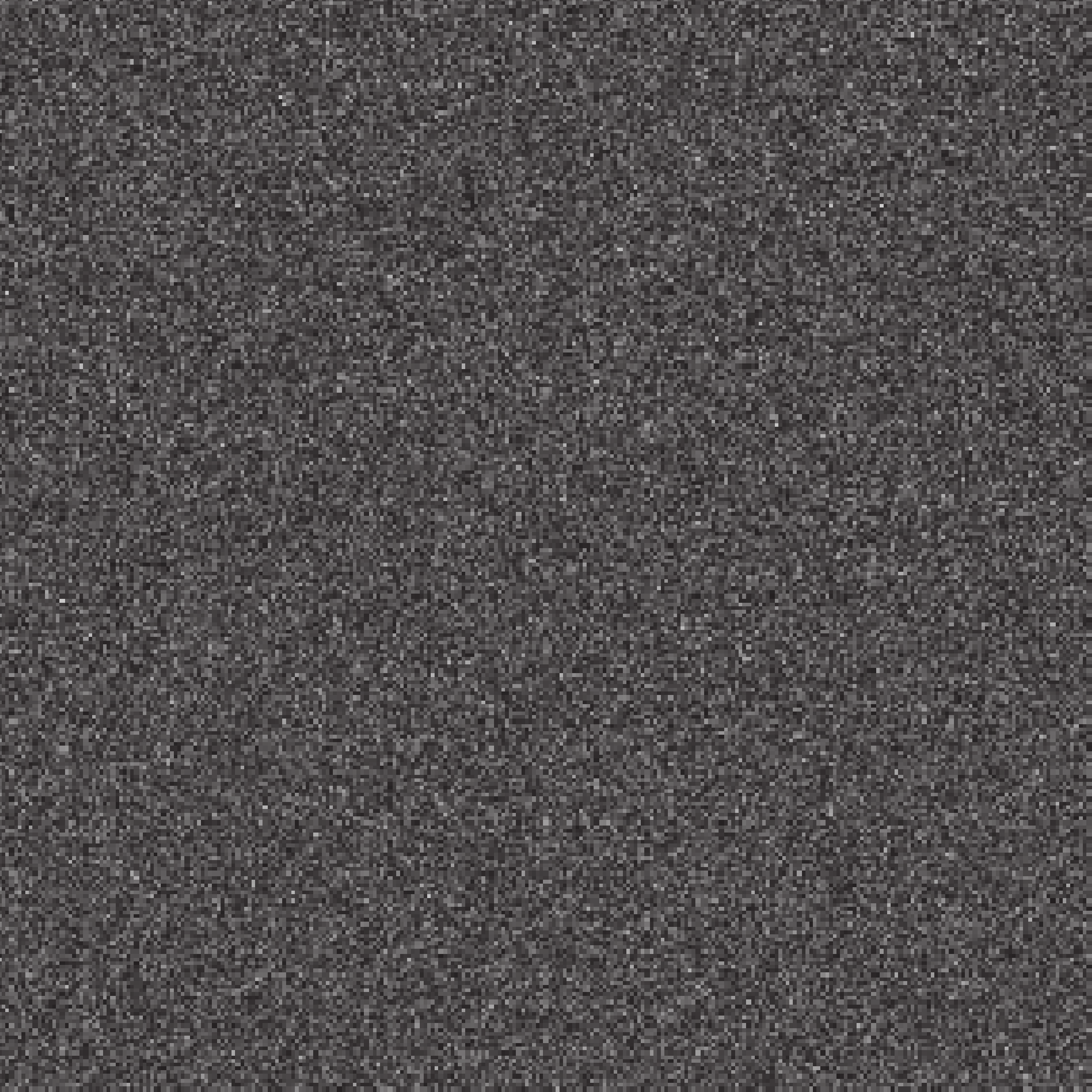}}
\subfigure[Difference between the 2D-FFAST reconstructed and original `Cal' images. Amplified by 5 for visualization. Notice there is zero error for the empty regions.]
{\label{fig:caln_spect_diff}
\includegraphics[width=.31\linewidth]{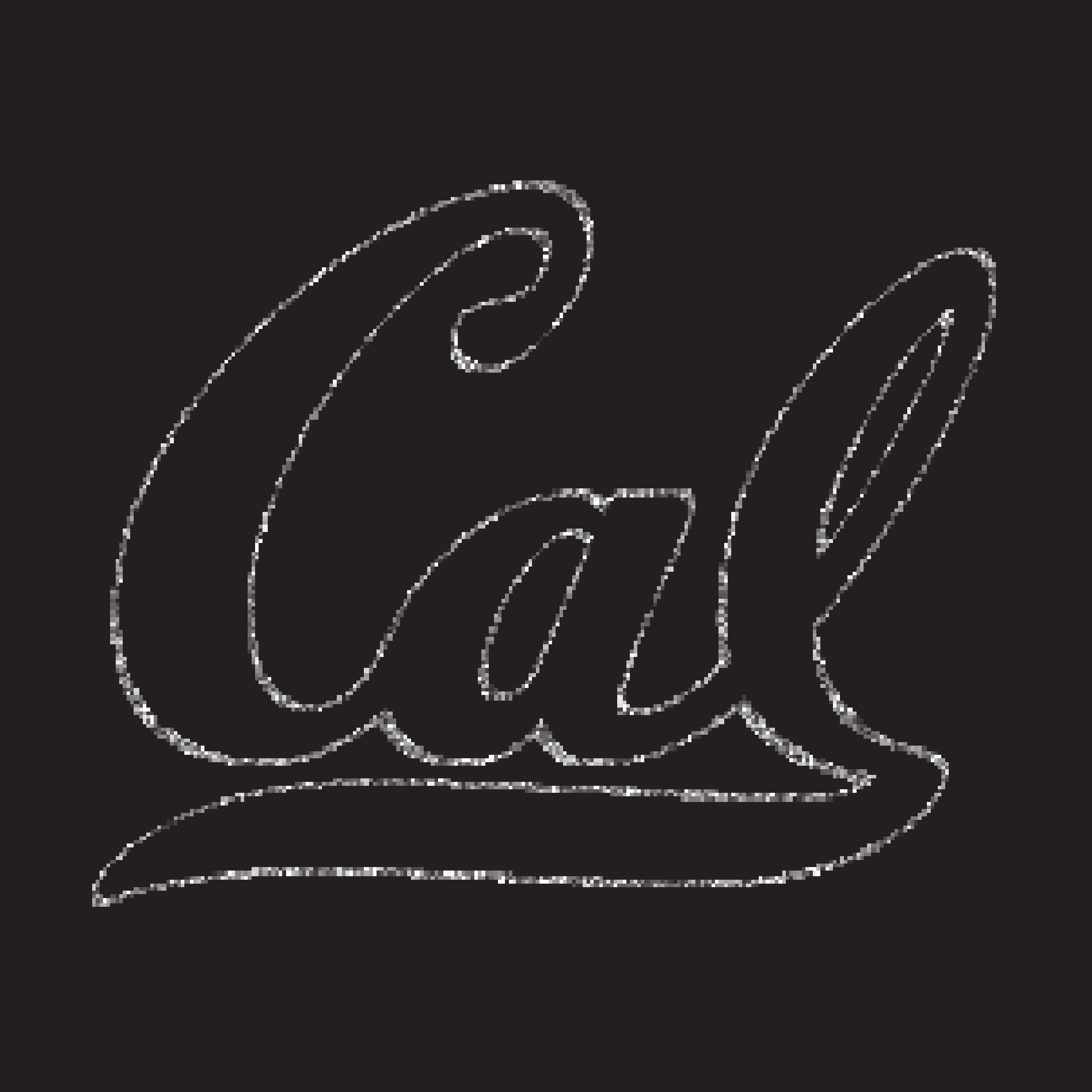}}

\caption{The figure shows the performance of the 2D FFAST algorithm for signals with exactly-sparse spatial image domain, with non-uniform support of the non-zero pixels and additive Gaussian noise.}
\end{figure}

In this section, we demonstrate the application of a noise robust variant of our algorithm 2D-R-FFAST to reconstructing the 'Cal' image from its noisy spatial samples. The 'Cal' image shown in Fig.~\ref{fig:caln_spect} is constructed by adding white Gaussian noise to the `Cal' image such that the SNR is $13$dB. We use a $3$ stage 2D-FFAST architecture with sampling periods $5,7$ and $8$ in each of the $3$ stages respectively. To increase noise robustness to 2D-FFAST, we increase the number of delays to $5$ per dimension. The number of measurements is $43200$. 

Figure \ref{fig:caln_spect_rec} shows the 2D-R-FFAST reconstructed image. The 2D-FFAST algorithm perfectly recovers the sparse coefficient locations in the DFT domain and results in a normalized mean-squared error of 0.0136. Figure \ref{fig:caln_diff} and \ref{fig:caln_spect_diff} show the amplified difference images for the noisy and 2D-FFAST reconstructed images respectively. As seen in the figures, 2D-R-FFAST reduces noise in regions where the pixels are zero in the original image as it utilizes the sparsity nature of the image.

\subsection{Application of the 2D-R-FFAST for MR imaging}\label{sec:sim_brain}
\begin{figure}[!ht]
\centering     
\subfigure[Log intensity plot of the 2D-DFT of the original `Brain' image. The red enclosed region is fully sampled and used for the stable inversion.]{\label{fig:dftfull}\includegraphics[width=.31\linewidth]{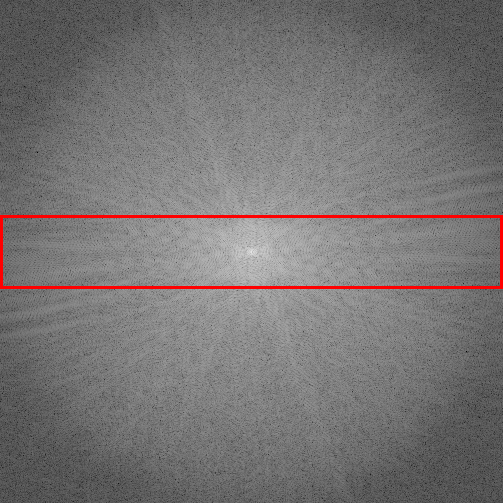}}
\subfigure[Original `Brain' image in spatial domain.]{\label{fig:spatialfull}\includegraphics[width=.31\linewidth]{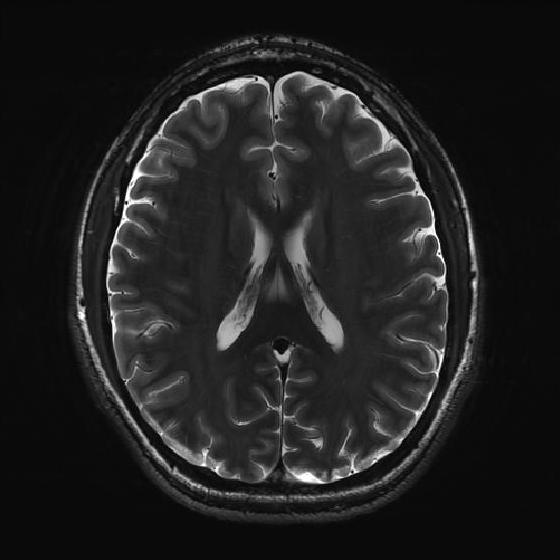}}
\subfigure[Reconstructed `Brain' image using the 2D-R-FFAST architecture along with the fully sampled center frequencies. The total number of Fourier measurements used is $60.18\%$.]{\label{fig:rfull}\includegraphics[width=.31\linewidth]{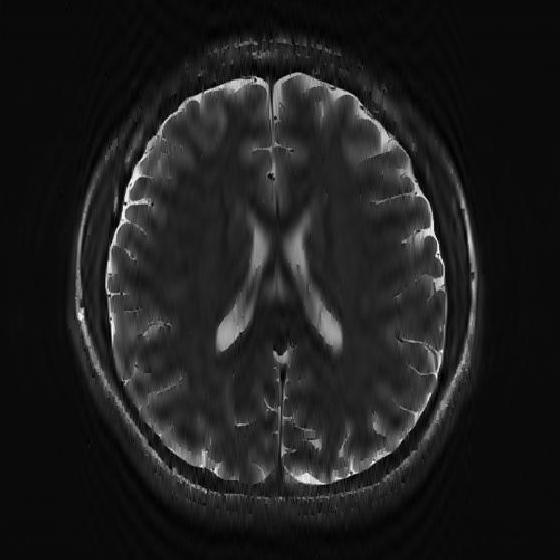}}
\subfigure[Log intensity plot of 2D-DFT of the original `Brain' image, after application of the vertical difference operation.]{\label{fig:dftdiff}\includegraphics[width=.31\linewidth]{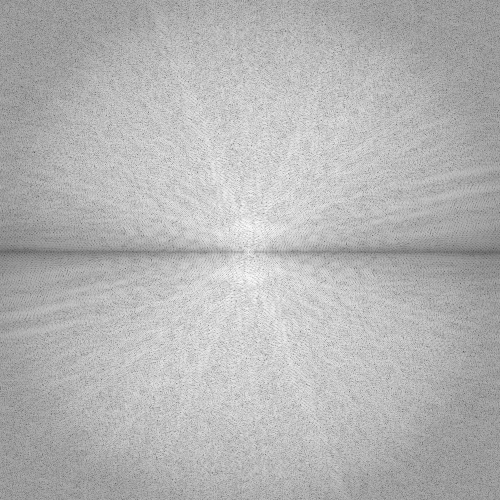}}
\subfigure[Differential `Brain' image obtained using the vertical difference operation on the original `Brain' image.]{\label{fig:spatialdiff}\includegraphics[width=.31\linewidth]{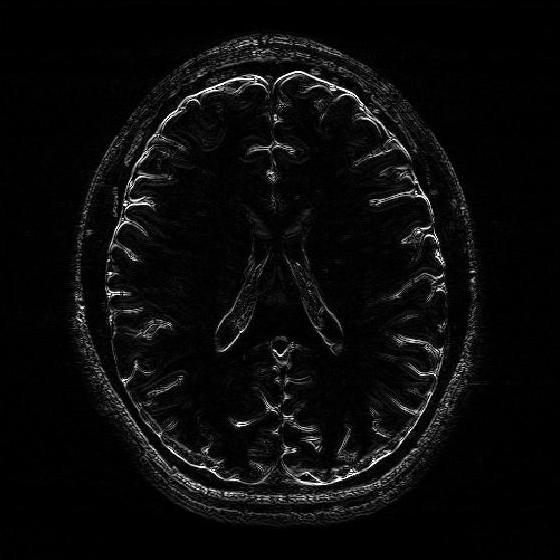}}
\subfigure[Differential `Brain' image reconstructed using the 2D-FFAST algorithm from $56.71\%$ of Fourier measurements.]{\label{fig:rdiff}\includegraphics[width=.31\linewidth]{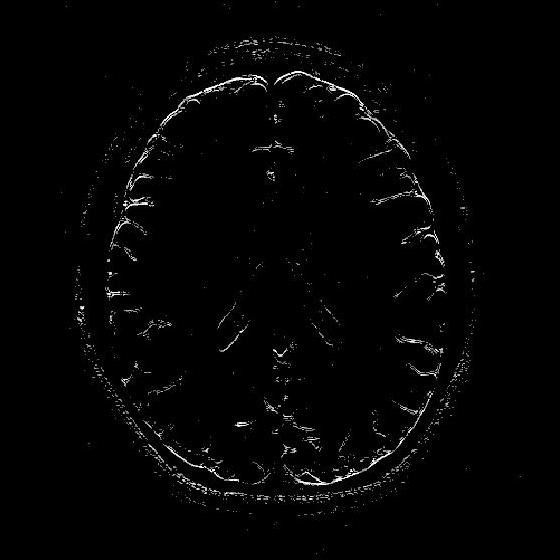}}
\caption{Application of the 2D-FFAST algorithm to reconstruct the `Brain' image acquired on an MR scanner with dimension $504 \times 504$. We first reconstruct the differential `Brain' image shown in Fig.~\ref{fig:spatialdiff}, using $\stages =3$ stage 2D-R-FFAST architecture with $15$ random delays in each of the $3$ stages of the 2D-FFAST architecture. Additionally we acquire all the Fourier measurements from the center frequency as shown, by the red enclosure, in Fig.~\ref{fig:dftfull}. Then, we do a stable inversion using the reconstructed differential `Brain' image of Fig.~\ref{fig:rdiff} and the fully sampled center frequencies of Fig.~\ref{fig:dftfull}, to get a reconstructed full `Brain' image as shown in Fig.~\ref{fig:rfull}. Our proposed two-step acquisition and reconstruction procedure takes overall $60.18\%$ of Fourier measurements.}
\end{figure}
In this section, we apply the 2D-R-FFAST algorithm to reconstruct a brain image acquired on an MR scanner with dimension $504 \times 504$. In MR imaging the measurements are acquired in the Fourier domain and the task is to reconstruct the spatial image from significantly less number of Fourier measurements~\cite{Lustig:2007cu}. To reconstruct the full brain image using 2D-R-FFAST, we perform the following two-step procedure: 
\begin{itemize}
\item {\em Differential space signal acquisition}: We perform a vertical finite difference operation on the image by multiplying the 2D-DFT signal with $1 - e^{ 2 \pi \imath \omega_x}$. This operation effectively creates an approximately sparse differential image, as shown in Fig.~\ref{fig:spatialdiff}, in spatial domain and can be reconstructed using 2D-FFAST. Note, that the finite difference operation can be performed on the sub-sampled data and at no point we access all the input Fourier measurements. The differential brain image is then sub-sampled and reconstructed using a $3$ stage 2D-FFAST architecture. Also, since the brain image is approximately sparse, we take $15$ delay sub-streams in each of the $3$ stages of the 2D-FFAST architecture, instead of $3$ delay sub-streams as in the exactly sparse case. The 2D-R-FFAST algorithm reconstructs the differential brain image using $56.71\%$ of Fourier measurements.
\item {\em Inversion using fully sampled center frequencies}: After reconstructing the differential brain image, as shown in Fig.~\ref{fig:rdiff}, we invert the finite difference operation by dividing the 2D-DFT measurements with $1 - e^{ 2 \pi \imath \omega_x}$. Since the inversion is not stable near the center of the Fourier domain, only the non-center frequencies are inverted. The center region of the 2D-DFT is fully sampled and used in the inversion process.
\item Overall we use a total of $60.18\%$ of Fourier measurements to reconstruct the brain image using the 2D-R-FFAST algorithm along with the fully sampled center frequencies. The resulting signal-to-noise ratio of the reconstructed image is  $4.5173$ dB. While the reconstruction error is not as good as state-of-the-art compressed sensing MRI results, we note that the 2D-R-FFAST has \em both low computational complexity and low sample complexity \em, which none of the state-of-the-art compressed sensing results in MRI can achieve. 
\end{itemize}

\section{Conclusion}\label{sec:conclusion}
We have shown that a 2D sparse DFT with both low sample complexity and low computational complexity can be practically achieved with the 2D-FFAST algorithm. We have made the connection between the 2D sparse DFT problem and the sparse graph decoding problem and provides theoretical basis for our algorithm. Our applications show promise for practical implementation in imaging systems and we believe additional incorporation of domain knowledge can further improve the 2D-FFAST performance.

\begin{appendices}

\section{Proof of Theorem~\ref{thm:main} for 2D-FFAST}\label{app:main}
\begin{figure}[!ht]
\begin{center}
\includegraphics[width=\linewidth]{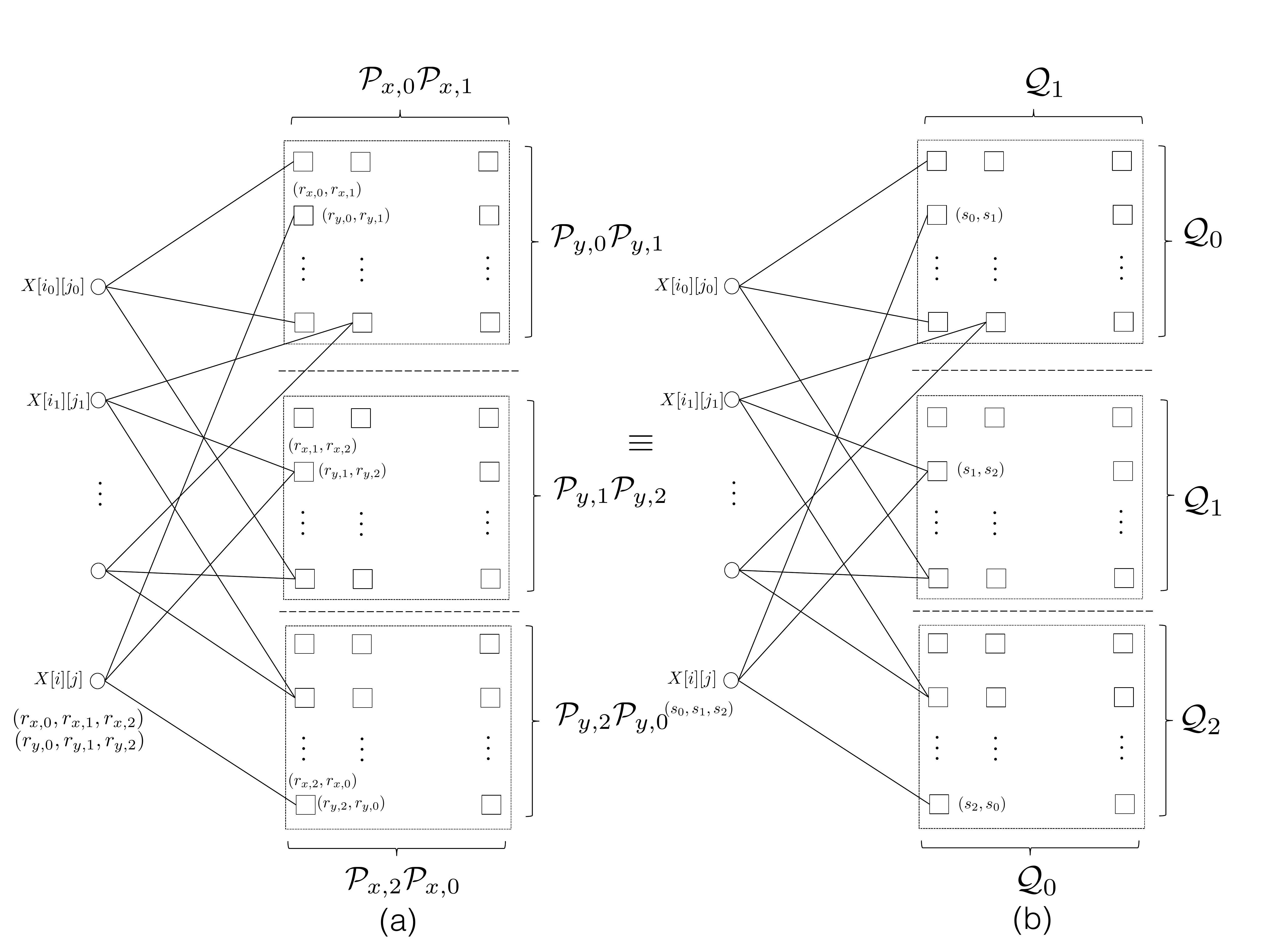}
\caption{A bi-partite graph representation of relation between the non-zero 2D-DFT coefficients, of a 2D signal $\msignal$, and the observations of a $3$-stage 2D-FFAST architecture. (a) A non-zero DFT coefficient with support $(i,j)$ is indexed by a $6$-tuple $((\rr{0},\rr{1},\rr{2}),(\rc{0},\rc{1},\rc{2}))$, where $\rr{\ell} = i \text{ mod } \primefr{\ell}$ and $\rc{\ell} = j \text{ mod } \primefc{\ell}$, for $\ell =0,1,2$. Each check node is indexed by a quadruplet, e.g., check node in stage $0$ is indexed by $((\rr{0},\rr{1}),(\rc{0},\rc{1}))$. A non-zero DFT coefficient with an index $((\rr{0},\rr{1},\rr{2}),(\rc{0},\rc{1},\rc{2}))$ is connected to a check node $((\rr{0},\rr{1}),(\rc{0},\rc{1}))$ in stage $0$. (b) A non-zero DFT coefficient with support $(i,j)$ is indexed by a triplet $(\sr{0},\sr{1},\sr{2})$, where $\sr{\ell} = \rr{\ell}\primef{\ell} + \rc{\ell}$, for $\ell = 0,1,2$. Each check node is indexed by a doublet, e.g., check node in stage $0$ is indexed by $(\sr{0},\sr{1})$. A non-zero DFT coefficient with an index $(\sr{0},\sr{1},\sr{2})$ is connected to a check node $(\sr{0},\sr{1})$ in stage $0$.}
\label{fig:sparsegraph2}
\end{center}
\end{figure}

In this section we provide a proof of Theorem~\ref{thm:main}. In \cite{Pawar:2013da}, we have shown that for 1D signals, for all values of the sparsity index $0 < \sindex < 1$, and sufficiently large $(k,N)$, where $\sparsity = O(N^{\sindex})$, there exists a 1D-FFAST architecture, that computes a $\sparsity$-sparse 1D-DFT $\transform$, using $O(\sparsity)$ sample in $O(\sparsity\log\sparsity)$ computations. The 2D-FFAST algorithm succeeds with probability approaching $1$, asymptotically in the number of measurements. In order to show a similar result for the 2D signals, we use the following approach:

We show that, for any given sparsity index $\sindex$ and sufficiently large $\sparsity, N=N_x \times N_y$, there exists a mapping between a 2D-FFAST architecture and a 1D-FFAST architecture designed for parameters $(\sindex,\sparsity,N)$. In particular, we show that a 2D-FFAST has a bi-partite graph representation corresponding to a bi-partite graph representation for a 1D-FFAST. The proof then follows from the results in \cite{Pawar:2013da}. For the sake of brevity we show the mapping for $\sindex = 2/3$. The mapping extends in a straightforward way for any $0 < \sindex < 1$.

%
%
%
%
%
%

Let $\{ \pprimef{0},  \pprimef{1},  \pprimef{2} \}$ be a set of pairwise co-prime integers such that $\pprimef{i} = \binsize + O(1) $ and let $\{\primefr{0},\primefr{1},\primefr{2}\}$ and $\{\primefc{0},\primefc{1},\primefc{2}\}$ be two sets of co-prime integers such that $\primefr{i}\primefc{i} = \pprimef{i}$. Also, let $N_x = \primefr{0}\primefr{1}\primefr{2}$, $N_y = \primefc{0}\primefc{1}\primefc{2}$, and $\sparsity = \primefr{0}\primefr{1}\primefc{0}\primefc{1}$. Then, $\sparsity = O(N^{2/3})$, where $N = N_xN_y$, that is, $\sindex = 2/3$. Consider a $3$-stage 2D-FFAST architecture with the sub-sampling factors $(\subr{0}, \subc{0}) = (\primefr{2}, \primefc{2})$, $(\subr{1}, \subc{1}) = (\primefr{0}, \primefc{0})$, $(\subr{2}, \subc{2}) = (\primefr{1}, \primefc{1})$ for the $3$ stages respectively. 


Using the CRT every 2D location $(0,0) \leq (i,j) < (N_x,N_y)$, is uniquely represented by $(\rr{0},\rr{1},\rr{2})$ and $(\rc{0},\rc{1},\rc{2})$, where $\rr{\ell} = i \text{ mod } \primefr{\ell}$ and $\rc{\ell} = j \text{ mod } \primefc{\ell}$, for $\ell = 0,1,2$. Then, a uniformly random choice of the 2D-support for a non-zero DFT coefficient corresponds to a uniformly random choice of the remainders $(\rr{0},\rr{1},\rr{2})$ and $(\rc{0},\rc{1},\rc{2})$. Using the sub-sampling-aliasing and the circular shift properties explained in Section~\ref{sec:extension} we obtain a sparse graph, shown in Fig.~\ref{fig:sparsegraph2}(a), representing the relation between the $\sparsity$ non-zero DFT coefficients of the 2D signal $\msignal$, and the bin-observations. Note that a non-zero DFT coefficient with a 2D support $(i,j)$ is connected to a check node indexed by $(\rr{0},\rr{1})$ and $(\rc{0},\rc{1})$ in stage $0$. 

Recall that each $\pprimef{\ell} = \primefr{\ell}\primefc{\ell}$ for $\ell = 0,1,2$. Then, any doublet $(0,0) \leq (i,j) < (\primefr{\ell},\primefc{\ell})$ can be uniquely represented by an integer $q = i\primefc{\ell} + j$. Using this mapping, we convert every doublet $(\rr{\ell},\rc{\ell})$ into a single number $\sr{\ell} = \rr{\ell}\primefc{\ell} + \rc{\ell}$, for $\ell = 0,1,2$. Thus, each check node can now be represented by a doublet instead of a quadruplet and each non-zero DFT coefficient can now be indexed by a triplet instead of a $6$-tuple. For example, a check node, in stage $0$, with a quadruplet index $((\rr{0},\rr{1}),(\rc{0},\rc{1}))$ can be represented by a doublet $(\sr{0},\sr{1})$. The re-arranged bi-partite graph with this new labeling is shown in Fig~\ref{fig:sparsegraph2}(b).

The bi-partite graph of Fig~\ref{fig:sparsegraph2}(b) has $\sparsity$ left nodes and $\pprimef{0}\pprimef{1} + \pprimef{1}\pprimef{2} +\pprimef{2}\pprimef{0}$ right nodes. Further, each left node has one uniformly random neighbor in each of the $3$ stages. This graph is a member of the ensemble of bi-partite graphs generated by a $3$-stage 1D-FFAST architecture, in \cite{Pawar:2013da}, designed for the following parameters: $\{\pprimef{0},\pprimef{1},\pprimef{2}\}$ is a set of co-prime integers such that $\pprimef{\ell} = O(\binsize)$, $N = \prod_{\ell=0}^2\pprimef{\ell}$ and $\sparsity = \pprimef{0}\pprimef{1}$, that is, $\sindex = 2/3$. Thus, the proof follows from the results of \cite{Pawar:2013da}, for $\sindex = 2/3$.
\eop

\section{Proof of Theorem~\ref{thm:main_r} for 2D-R-FFAST}\label{app:main_r}

In \cite{Pawar:2014gx}, we have shown that a noiseless 1D-FFAST framework can be made {\em noise-robust} by using $O(\log N)$, where $ N$ is an ambient signal dimension,  number of delay-chains per sub-sampling stage. Also, as explained in Section~\ref{sec:extension}, for 2D-DFT, we follow a {\em separable} approach and perform $2$ independent ratio-tests to determine the 2D support of a single-ton bin. Hence, all the results from \cite{Pawar:2014gx} follow and the noise robust 2D-FFAST algorithm computes $\sparsity$-sparse 2D-DFT in sub-linear time of $O(\sparsity\log^{4} N)$ using $O(\sparsity\log^{3}N)$ measurements.

\end{appendices}

\endgroup

\bibliographystyle{IEEEtran}
\bibliography{paper}
\end{document}